\documentclass[sigconf, nonacm]{acmart}

\newcommand\vldbavailabilityurl{URL_TO_YOUR_ARTIFACTS}

\usepackage{xspace}
\usepackage{algorithmicx}
\usepackage{subcaption,graphicx,url,multirow,multicol,mathtools}
\usepackage[noend]{algpseudocode}
\usepackage{balance, afterpage}
\usepackage{paralist}
\usepackage{booktabs}
\usepackage{color}
\usepackage{arydshln}
\usepackage{enumitem}
\usepackage[most]{tcolorbox}
\usepackage{stmaryrd,amsthm} 
\usepackage{bm}

\newcommand{\myparagraph}[1]{\vspace{1pt}\noindent {\bf #1.}}

\mathchardef\mhyphen="2D

\newcommand{\whp}[1]{w.h.p.}
\newcommand{\defn}[1]{{\bf{\emph{#1}}}}

\usepackage[ruled]{algorithm}

\newcommand{\BigO}[1]{\ensuremath{O(#1)}}
\newcommand{\BigOmega}[1]{\ensuremath{\Omega(#1)}}

\makeatletter
\newcommand{\newreptheorem}[2]{\newtheorem*{rep@#1}{\rep@title}\newenvironment{rep#1}[1]{\def\rep@title{#2 \ref*{##1}}\begin{rep@#1}}{\end{rep@#1}}}
\makeatother

\newreptheorem{theorem}{Theorem}
\newreptheorem{lemma}{Lemma}

\algblock{ParFor}{EndParFor}
\algnewcommand\algorithmicparfor{\textbf{parfor}}
\algnewcommand\algorithmicpardo{\textbf{do}}
\algnewcommand\algorithmicendparfor{}
\algrenewtext{ParFor}[1]{\algorithmicparfor\ #1\ \algorithmicpardo}
\algrenewtext{EndParFor}{\algorithmicendparfor}
\algtext*{EndParFor}

\algdef{SE}[DOWHILE]{Do}{doWhile}{\algorithmicdo}[1]{\algorithmicwhile\ #1}%

\algdef{SE}[SUBALG]{Indent}{EndIndent}{}{\algorithmicend\ }%
\algtext*{Indent}
\algtext*{EndIndent}

\definecolor{ao}{rgb}{0.0, 0.5, 0.0}

\newcommand{\algprefix}{Algorithm\xspace}
\newcommand{\algname}[1]{\textnormal{\textsc{#1}}}

\newcommand{\ournd}{\algname{arb-nucleus-decomp}\xspace}
\newcommand{\ourcountrec}{\algname{rec-list-cliques}\xspace}
\newcommand{\sariyuce}{Sariy{\"{u}}ce\xspace}
 
\newcommand{\cnucleus}[3]{$#1$-$(#2, #3)$ nucleus} 
\newcommand{\nucleusdecomp}[2]{$(#1, #2)$ nucleus decomposition}

\usepackage{enumitem}
\setlist{noitemsep,topsep=0pt,parsep=0pt,partopsep=0pt}
\usepackage{microtype} 

\definecolor{change}{rgb}{0,0,0}

\usepackage{footnote}
\makesavenoteenv{algorithm}

\newcommand{\intheapp}{appendix\xspace}

\makeatletter
\renewenvironment{proof}[1][\proofname]{\par
  \vspace{-0.5\topsep}
  \pushQED{\qed}%
  \normalfont
  \topsep0pt \partopsep0pt 
  \trivlist
  \item[\hskip\labelsep
        \scshape
    #1\@addpunct{.}]\ignorespaces
}{%
  \popQED\endtrivlist\@endpefalse
  \addvspace{4pt} 
}
\makeatother

\begin{document}
\title[Theoretically and Practically Efficient Parallel Nucleus Decomposition]{Theoretically and Practically Efficient \\ Parallel Nucleus Decomposition}


\author{Jessica Shi}
\affiliation{\institution{MIT CSAIL}}
\email{jeshi@mit.edu}
\author{Laxman Dhulipala}
\affiliation{\institution{MIT CSAIL}}
\email{laxman@mit.edu}
\author{Julian Shun}
\affiliation{\institution{MIT CSAIL}}
\email{jshun@mit.edu}

\begin{abstract}
This paper studies the nucleus decomposition problem, which has been shown to be useful in finding dense substructures in graphs. 
We present a novel parallel algorithm that is efficient both in theory and in practice. Our algorithm achieves a work complexity matching the best sequential algorithm while also having low depth (parallel running time), which significantly improves upon the only existing parallel nucleus decomposition algorithm (\sariyuce \textit{et al.}, PVLDB 2018). 
The key to the theoretical efficiency of our algorithm is 
the use of a theoretically-efficient parallel algorithms for clique listing and bucketing. 
We introduce several new practical optimizations, including a new multi-level hash table structure to store information on cliques space-efficiently and a technique for traversing this structure cache-efficiently. 
On a 30-core machine with two-way hyper-threading on real-world graphs, we achieve up to a 55x speedup over the state-of-the-art parallel nucleus decomposition algorithm by \sariyuce \textit{et al.}, and up to a 40x self-relative parallel speedup.  We are able to efficiently compute larger nucleus decompositions than prior work on several million-scale graphs for the first time.

\end{abstract}

\maketitle

\ifdefempty{\vldbavailabilityurl}{}{
\vspace{-5pt}
\begingroup\small\noindent\raggedright\textbf{PVLDB Artifact Availability:}\\
The source code, data, and/or other artifacts have been made available at \url{https://github.com/jeshi96/arb-nucleus-decomp}.
\endgroup
}
\section{Introduction}

Discovering dense substructures in graphs is a fundamental topic in graph 
mining, and has been studied across many areas including
computational biology~\cite{bader2003automated,fratkin2006motifcut},
spam and fraud-detection~\cite{gibson2005discovering}, and 
large-scale network analysis~\cite{angel2014dense}.
Recently, \sariyuce{} \textit{et al.}~\cite{Sariyuce2017} introduced the nucleus decomposition problem, which
generalizes the influential notions of $k$-cores and $k$-trusses to $k$-$(r,s)$ nucleii,
and can better capture higher-order structures in the graph. 
Informally, a $k$-$(r,s)$ nucleus is the maximal induced subgraph such that every
$r$-clique in the subgraph is contained in at least $k$ $s$-cliques.
The goal of the \nucleusdecomp{r}{s} problem is to identify for each $r$-clique in the graph, the
largest $k$ such that it is in a $k$-$(r,s)$ nucleus.

Solving the \nucleusdecomp{r}{s} problem is a significant computational challenge for several reasons.
First, simply counting and enumerating $s$-cliques is a challenging task, even for
modest $s$. 
Second, storing information for all $r$-cliques can require a large amount of space, even 
for relatively small graphs.
Third, engineering fast and high-performance solutions to this problem requires taking
advantage of parallelism due to the computationally-intensive nature of listing cliques.
{\color{change} There are two well-known parallel paradigms for approaching the \nucleusdecomp{r}{s} problem, a global peeling-based model and a local update model that iterates until convergence~\cite{sariyuce2017parallel}. The former is inherently challenging to parallelize due to sequential dependencies and necessary synchronization steps~\cite{sariyuce2017parallel}, which we address in this paper, and we demonstrate that the latter requires orders of magnitude more work to converge to the same solution and is thus less performant. }

Lastly, it is unknown whether existing sequential and parallel algorithms for this
problem are theoretically efficient. 
Notably, existing algorithms perform more work than the
fastest theoretical algorithms for $k$-clique enumeration on sparse graphs~\cite{ChNi85,shi2020parallel}, and it is open whether
one can solve the \nucleusdecomp{r}{s} problem in the same work as $s$-clique enumeration.

In this paper, we design a novel parallel algorithm for the nucleus decomposition problem.
We address the computational challenges by designing 
a theoretically efficient parallel algorithm for \nucleusdecomp{r}{s} that nearly matches the work for $s$-clique enumeration, along with new techniques that improve the space
and cache efficiency of our solutions.
The key to our theoretical efficiency is combining a combinatorial lemma bounding the total sum over all $k$-cliques in the graph of the minimum degree vertex in this clique~\cite{EdRoSe20} with a theoretically efficient $k$-clique listing algorithm~\cite{shi2020parallel},
enabling us to to provide a strong upper bound on the overall work of our
algorithm.\footnote{In the version of our paper appearing in the \emph{International Conference on Very Large Databases (VLDB)}, 2022~\cite{vldbversion}, we claimed the combinatorial lemma as a new contribution. However, we were later made aware that the lemma first appears in Eden \textit{et al.}'s work~\cite{EdRoSe20}.}
As a byproduct, we also obtain the most theoretically-efficient serial algorithm for \nucleusdecomp{r}{s}.
We provide several new optimizations for improving the practical efficiency of our algorithm, including a new multi-level hash table structure to space efficiently store data associated with cliques, a technique for efficiently traversing this structure in a cache-friendly manner, and methods for reducing contention and further reducing space usage. 
Finally, we experimentally study our parallel algorithm on various real-world graphs and $(r,s)$ values, and find that it achieves
between 3.31--40.14x self-relative speedup on a 30-core machine with two-way hyper-threading.
The only existing parallel algorithm for nucleus decomposition is by \sariyuce{} \textit{et al.}~\cite{sariyuce2017parallel}, but their algorithm requires much more work than the best sequential algorithm. Our algorithm achieves
between 1.04--54.96x speedup over the state-of-the-art parallel
nucleus decomposition of \sariyuce{} \textit{et al.}, and our algorithm can scale to larger $(r,s)$ values, due to our improved theoretical efficiency and our proposed optimizations. We are able to compute
the $(r,s)$ nucleus decomposition for $r>3$ and $s>4$ on several million-scale
graphs for the first time.
We summarize our contributions below:
\begin{itemize}[topsep=1pt,itemsep=0pt,parsep=0pt,leftmargin=10pt]
  \item The first theoretically-efficient parallel algorithm for the nucleus decomposition problem.
  \item A collection of practical optimizations that enable us to design a fast implementation of our algorithm.
  \item Comprehensive experiments showing that our new algorithm achieves
  up to a 55x speedup over the state-of-the-art algorithm 
  by \sariyuce{} \textit{et al.}, and up to a 40x self-relative parallel speedup on a 30-core machine with two-way hyper-threading.
\end{itemize}

\section{Related Work}

The nucleus decomposition problem is inspired by and closely related to the $k$-core problem, which was defined independently by Seidman~\cite{seidman83network}, and by Matula and Beck~\cite{matula83smallest}.
The $k$-core of a graph is the maximal subgraph of the graph where the induced degree 
of every vertex is at least $k$.
The \emph{coreness} of a vertex is the maximum value of $k$ such that the vertex participates in a $k$-core.
Matula and Beck provided a linear time algorithm based on peeling vertices that computes the coreness value of all vertices~\cite{matula83smallest}.

In subsequent years, many concepts capturing dense near-clique substructures were proposed, including $k$-trusses (or triangle-cores), $k$-plexes~\cite{seidmanfoster}, and $n$-clans and $n$-clubs~\cite{mokken1979cliques}.
In particular, $k$-trusses were proposed independently by Cohen~\cite{cohen2008trusses}, Zhang \textit{et al.}~\cite{zhang2012extracting}, and Zhou \textit{et al.}~\cite{zhao2012large} with the goal of efficiently obtaining dense clique-like substructures.
Unlike other near-clique substructures like $k$-plexes, $n$-clans, and $n$-clubs, which are computationally intractable to enumerate and count, $k$-trusses can be efficiently found in polynomial-time.
Many parallel, external-memory, and distributed algorithms have been developed in the past decade for $k$-cores~\cite{montresor2012distributed, farach2014computing, khaouid2015k, Kabir2017, DhBlSh17,Wen2019} and {\color{change}$k$-trusses~\cite{Wang2012, chen2014distributed, Zou16, kabir2017parallel, smith2017truss, ChLaSuWaLu20,CoDeGrMaVe18, LoSpKuSrPoFrMc18,  BlLoKi19}}, and computing all trussness values of a graph is one of the challenge problems in the yearly MIT GraphChallenge~\cite{samsi2017static}.
A related problem is to compute the $k$-clique densest subgraph~\cite{Tsourakakis15} and $(k,\Psi)$-core~\cite{FaYuChLaLi19}, for which efficient parallel algorithms have been recently designed~\cite{shi2020parallel}.
The concept of a \nucleusdecomp{r}{s} was first proposed by \sariyuce{} \textit{et al.}\ as a principled approach to discovering dense substructures in graphs that generalizes $k$-cores and $k$-trusses~\cite{Sariyuce2017}. They also proposed an algorithm for efficiently finding the hierarchy associated with a \nucleusdecomp{r}{s}~\cite{SaPi16}.
\sariyuce{} \textit{et al.}\ later proposed parallel algorithms for nucleus decomposition based on local computation~\cite{sariyuce2017parallel}.
Recent work has studied nucleus decomposition in probabilistic graphs~\cite{esfahani2020nucleus}.

Clique counting and enumeration are fundamental subproblems required for computing nucleus decompositions. 
A trivial algorithm enumerates $c$-cliques in $O(n^{c})$ work, and using a thresholding argument improves the work for counting to $O(m^{c/2})$~\cite{AYZ97}. 
The current fastest combinatorial algorithms for $c$-clique enumeration for sparse graphs are based on the seminal results of Chiba and Nishizeki~\cite{ChNi85}, who show that all $c$-cliques can be enumerated in $O(m\alpha^{c-2})$ where $\alpha$ is the arboricity of the graph. 
We defer to the survey of Williams for an overview of theoretical algorithms for this problem~\cite{vassilevska2009efficient}. 
The current state-of-the-art practical algorithms for $k$-clique counting are all based on the Chiba-Nishizeki algorithm~\cite{Danisch18, shi2020parallel, li2020ordering}.

Researchers have also studied $k$-core-like computations in 
bipartite graphs~\cite{Shi2020,Wang2020,lakhotia20receipt,SariyuceP18, liu2020alphabeta}, as well as how to maintain $k$-cores and $k$-trusses in dynamic graphs~\cite{Li2014, Wen2019,Zhang2017,sariyuce2016incremental,Akbas2017,Huang2014, aridhi2016distributed, lin2021dynamickcore,Luo2021,Hua2020,SCS20,Jin2018,Luo2020,Zhang2019a,Huang2015}.
Very recently, \sariyuce{} proposed a motif-based decomposition, which generalizes 
the connection between $r$-cliques and $s$-cliques in nucleus decomposition to any pair of subgraphs~\cite{sariyuce2021motif}.

\section{Preliminaries}\label{sec:prelims}
\myparagraph{Graph Notation and Definitions} We consider graphs $G =(V,E)$ to be
simple and undirected, where $n = |V|$ and $m = |E|$. For analysis, we assume $m=\BigOmega{n}$.
For vertices $v \in V$, we denote by $\text{deg}(v)$ the degree of $v$, and we denote by $N_G(v)$ the neighborhood of $v$ in $G$. If the graph is unambiguous, we let $N(v)$ denote the neighborhood of $v$. For a directed graph $DG$, $N(v) = N_{DG}(v)$ denotes the out-neighborhood of $v$.
The \defn{arboricity ($\bm{\alpha}$)} of a graph is the minimum number
of spanning forests needed to cover the graph. In general, $\alpha$ is
upper bounded by $\BigO{\sqrt{m}}$ and lower bounded by
$\BigOmega{1}$~\cite{ChNi85}.

A \defn{\cnucleus{c}{r}{s}} is a maximal subgraph $H$ of an undirected graph formed by the union of $s$-cliques $C_s$, such that each $r$-clique $C_r$ in $H$ has induced $s$-clique degree at least $c$ (i.e., each $r$-clique is contained within at least $c$ induced $s$-cliques). 
The \defn{\nucleusdecomp{r}{s} problem} is to compute all non-empty $(r,s)$-nuclei. Our algorithm outputs the \defn{$(r,s)$-clique core number} of each $r$-clique $C_r$, or the maximum $c$ such that $C_r$ is contained within a \cnucleus{c}{r}{s}.\footnote{The original definition of  \nucleusdecomp{r}{s} is stricter, in that it additionally requires any two $r$-cliques in the maximal subgraph $H$ to be connected via $s$-cliques~\cite{Sariyuce2017,SaPi16}. This requires additional work to partition the $r$-cliques, which the previous parallel algorithm~\cite{sariyuce2017parallel} does not perform, and is also out of our scope of this paper.}
The $k$-core and $k$-truss problems correspond to the $k$-$(1,2)$ and $k$-$(2,3)$ nucleus, respectively.

\myparagraph{Graph Storage} 
For theoretical analysis, we assume that our graphs are represented
in an adjacency hash table, where each vertex is associated with a
parallel hash table of its neighbors. 
In practice, we store
graphs in compressed sparse row (CSR) format.

\myparagraph{Model of Computation} 
We use the fundamental work-span model for our theoretical analysis, which is widely used in analyzing shared-memory parallel algorithms~\cite{JaJa92,CLRS}, with many recent practical uses~\cite{ShunRFM16,Dhulipala2020, sun2019supporting, wang2019pardbscan}.
The \defn{work} $W$ of an algorithm is the total number of operations, and the \defn{span} $S$ of an algorithm is the longest dependency path. Brent's scheduling theorem~\cite{Brent1974} upper bounds the parallel running time of an algorithm by $W/P + S$ where $P$ is the number of processors. A randomized work-stealing scheduler,
such as that in Cilk~\cite{BlLe99}, can be used in practice to obtain this running time in expectation.
The goal of our work is to develop \defn{work-efficient} parallel algorithms under this model, or algorithms with a work complexity that asymptotically matches the best-known sequential time complexity for the given problem. 
We assume that this model supports concurrent reads, concurrent writes, compare-and-swaps, atomic adds, and fetch-and-adds in $O(1)$ work and span. 

\myparagraph{Parallel Primitives} We use the following parallel primitives in our algorithms. 
Parallel \defn{prefix sum} takes as input a sequence $A$
of length $n$, an identity $\varepsilon$, and an associative binary
operator $\oplus$, and returns the sequence $B$ of length $n$ where
$B[i] = \bigoplus_{j < i} A[j] \oplus \varepsilon$. Parallel \defn{filter}
takes as input a sequence $A$ of length $n$ and a predicate function
$f$, and returns the sequence $B$ containing all $a \in A$ such that
$f(a)$ is true, maintaining the order that elements appeared in $A$. Both algorithms take $\BigO{n}$ work and $\BigO{\log n}$ span~\cite{JaJa92}. {\color{change} These primitives are basic sequence operations that represent essential building blocks of efficient parallel algorithms. They are implemented in practice using efficient parallel divide-and-conquer subroutines.} 

We also use \defn{parallel hash tables} that support insertion, deletion, and membership queries, and that can perform $n$ operations in $O(n)$
work and $O(\log n)$ span with high probability (\whp{})~\cite{gil91a}.\footnote{We
  say $O(f(n))$ \defn{with high probability (\whp{})} to indicate
  $O(cf(n))$ with probability at least $1-n^{-c}$ for $c \geq 1$,
  where $n$ is the input size.}  Given
two parallel hash tables $\mathcal{T}_1$ and $\mathcal{T}_2$ of size $n_1$ and
$n_2$ respectively, the intersection $\mathcal{T}_1 \cap
\mathcal{T}_2$ can be computed in $\BigO{\min(n_1,n_2)}$ work and
$\BigO{\log (n_1+n_2)}$ span \whp{}
For multiple hash tables $\mathcal{T}_i$ for $i \in [m]$, each of size $n_i$, the intersection $\bigcap_{i \in [m]} \mathcal{T}_i$ can be computed in $\BigO{\min_{i \in [m]}(n_i)}$ work and $\BigO{\log (\sum_{i \in [m]} n_i)}$ span \whp{}
These subroutines are essential to our clique counting and listing subroutines, allowing us to find the intersections of adjacency lists of vertices. We also use parallel hash tables in our algorithms as efficient data structures for $r$-clique access and aggregation.

\myparagraph{Parallel Bucketing} A \defn{parallel bucketing structure}
maintains a mapping from identifiers to buckets, which we use to group
$r$-cliques according to their incident $s$-clique counts. The
bucket value of identifiers can change, and the structure can efficiently update these buckets. 
We take identifiers to be values associated with
$r$-cliques, and use the structure to repeatedly extract all $r$-cliques in
the minimum bucket to process, which can cause the bucket
values of other $r$-cliques to change (other $r$-cliques that share vertices with extracted
$r$-cliques in our algorithm). 
Theoretically, the batch-parallel Fibonacci heap by Shi
and Shun~\cite{Shi2020} can be used to implement a bucketing structure storing $n$ objects that supports $k$ bucket insertions in $O(k)$ amortized expected
work and $O(\log n)$ span \whp{}, $k$ bucket update operations in
$O(k)$ amortized expected work and $O(\log^{2} n)$ span \whp{}, and extracts
the minimum bucket in $O(\log n)$ amortized expected work and $O(\log
n)$ span \whp{}
In our $(r, s)$ nucleus decomposition algorithm, we must extract the bucket of $r$-cliques with the minimum $s$-clique count in every round, and the work-efficiency of this Fibonacci heap contributes to our work-efficient bounds. 
However, our implementations use the bucketing structure by Dhulipala et
al.~\cite{DhBlSh17}, which we found to be more efficient in practice.  This structure obtains performance improvements by only materializing a constant number of the lowest buckets, reducing the number of times each $r$-clique's bucket must be updated. In retrieving new buckets, the structure also skips over large ranges of empty buckets containing no $r$-cliques, allowing for fast retrieval of the minimum non-empty bucket.

\myparagraph{$O(\alpha)$-Orientation} We use in our algorithms the parallel $c$-clique counting and listing algorithm by Shi \textit{et al.}~\cite{shi2020parallel}, which relies on directing a graph using a low out-degree orientation in order to reduce the amount of work that must be performed to find $c$-cliques. 
An \defn{$a$-orientation} of an undirected graph is a total ordering on the vertices such that when edges in the graph are directed from vertices lower in the ordering to vertices higher in the ordering, the out-degree of each vertex is bounded by $a$.
Shi \textit{et al.} provide parallel work-efficient algorithms to obtain an $O(\alpha)$-orientation, namely the parallel Barenboim-Elkin algorithm which takes $O(m)$ work and $O(\log^2 n)$ span, and the parallel Goodrich-Pszona algorithm which takes $O(m)$ work and $O(\log^2 n)$ span \whp{}
Besta et al.~\cite{Besta2020} present a parallel $O(\alpha)$-orientation algorithm that takes $O(m)$ work and $O(\log^2n)$ span.

\section{$(r,s)$ nucleus decomposition}
We present here our parallel work-efficient $(r,s)$ nucleus decomposition algorithm. Importantly, we introduce new theoretical bounds for $(r,s)$ nucleus decomposition, which also improve upon the previous best sequential bounds. We discuss in Section \ref{sec:alg-prelims} a key $s$-clique counting subroutine, and we present \ournd{}, our parallel $(r, s)$ nucleus decomposition algorithm, in Section \ref{sec:alg-nd}.

\subsection{Recursive $s$-clique Counting Algorithm} \label{sec:alg-prelims}

We first introduce an important subroutine, \ourcountrec{}, based on previous work from Shi \textit{et al.}~\cite{shi2020parallel}, which recursively finds and lists $c$-cliques in parallel. 
This subroutine is based on a state-of-the-art $c$-clique listing algorithm, which in practice balances performance and memory efficiency, outperforming other baselines, particularly for large graphs with hundreds of billions of edges~\cite{shi2020parallel}. 
This subroutine has been modified from previous work to integrate in our parallel $(r,s)$ nucleus decomposition algorithm, to both count the number of $s$-cliques incident on each $r$-clique, and  update the $s$-clique counts after peeling subsets of $r$-cliques. The main idea for \ourcountrec{} is to iteratively grow each $c$-clique by maintaining at every step a set of candidate vertices that are neighbors to all vertices in the $c$-clique so far, and prune this set as we add more vertices to the $c$-clique.

The pseudocode for the algorithm is shown in Algorithm~\ref{alg:count}.
\ourcountrec{} takes as input a directed graph $DG$, a set $I$ of potential neighbors to complete the clique (which for $c$-clique listing is initially $V$), the recursive level {\color{change}$r\ell$ (which for $c$-clique listing is initially set to $c$)}, a set $C$ of vertices in the clique so far (which for $s$-clique listing is initially empty), and a function $f$ to apply to each discovered $c$-clique. The directed graph $DG$ is an $O(\alpha)$-orientation of the original undirected graph, which allows us to reduce the work required for computing intersections.
\ourcountrec{} then uses repeated intersections on the set $I$ and the directed neighbors of each vertex in $I$ to find valid vertices to add to the clique $C$ (Line \ref{line:intersect}), and recurses on the updated clique (Line \ref{line:rec-count-recurse}). At the final recursive level, \ourcountrec{} applies the user-specified function $f$ on each discovered clique (Line \ref{line:base-case}).
Assuming that $DG$ is an $O(\alpha)$-oriented graph, and excluding the time required to obtain $DG$, \ourcountrec{} can perform $c$-clique listing in $O(m\alpha^{c-2})$ work and $O(c \log n)$ span \whp{}~\cite{shi2020parallel}.

\begin{algorithm}[!t]
  \footnotesize
 \begin{algorithmic}[1]
 \Procedure{rec-list-cliques}{$DG$, $I$, {\color{change}$r\ell$}, $C$, $f$}
 \State \Comment{$DG$ is the directed graph, $I$ is the set of potential neighbors to complete the clique, {\color{change}$r\ell$} is the recursive level, $C$ is the set of vertices in the clique so far, and $f$ is the desired function to apply to $c$-cliques.}
  \If{{\color{change}$r\ell = 1$}}
   \ParFor{$v$ in $I$}
   \State Apply $f$ on the $c$-clique $C \cup \{v\}$ \label{line:base-case}
   \EndParFor
    \State   \Return
 \EndIf
\ParFor{$v$ in $I$}
  \State $I' \leftarrow$ \algname{intersect}($I$, $N_{DG}(v)$) \Comment{Intersect $I$ with directed neighbors of $v$} \label{line:intersect}
  \State \ourcountrec{}($DG$, $I'$, {\color{change}$r\ell-1$}, $C \cup \{v\}$, $f$) \Comment{Add $v$ to the $c$-clique and recurse}\label{line:rec-count-recurse}
\EndParFor
\EndProcedure
 \end{algorithmic}
\caption{Parallel $c$-clique listing algorithm.} \label{alg:count}
\end{algorithm}

\subsection{$(r,s)$ Nucleus Decomposition Algorithm}\label{sec:alg-nd}
We now describe our parallel nucleus decomposition algorithm, \ournd{}.
\ournd{} computes the $(r,s)$ nucleus decomposition by first computing
and storing the  incident $s$-clique counts of each $r$-clique.
It then proceeds in rounds, where in each round, it peels, or implicitly removes, the $r$-cliques 
with the minimum $s$-clique counts. It updates the $s$-clique counts 
of the remaining unpeeled $r$-cliques, by decrementing the count by 1 for each $s$-clique 
that the unpeeled $r$-clique shares peeled $r$-cliques with. \ournd{} uses 
\ourcountrec{} as a subroutine, to compute and update $s$-clique counts.

\begin{figure}[t]
    \centering
   \begin{subfigure}{.36\textwidth}
   \centering
   \includegraphics[width=\columnwidth, page=1]{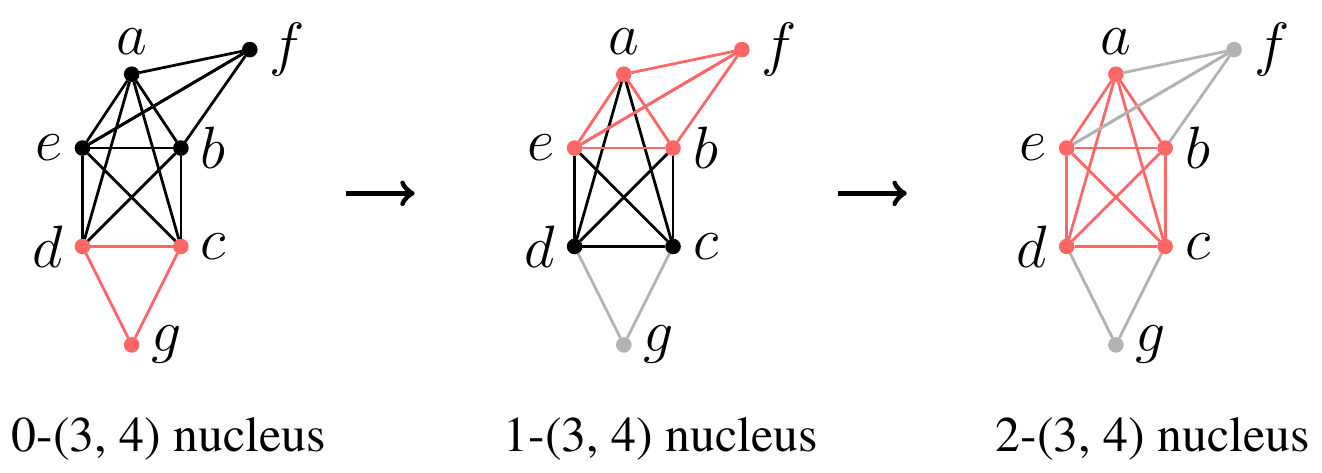}
   \end{subfigure}%
   \hfill
   \begin{subfigure}{.19\textwidth}
   \centering
   \includegraphics[width=\columnwidth, page=2]{figures/fig_nd_font.pdf}
      \end{subfigure}
   \caption{An example of our parallel nucleus decomposition algorithm \ournd{} for $(r,s) = (3,4)$. At each step, we peel in parallel all triangles (3-cliques) with the minimum 4-clique count; the vertices and edges that compose of these triangles  are highlighted in red. We then recompute the 4-clique count on the remaining triangles. Vertices and edges that no longer participate in any active triangles, due to previously peeled triangles, are shown in gray. 
   Each step is labeled with the $k$-(3,4) nucleus discovered, where $k$ is the 4-clique count of the triangles highlighted in red. The figure below labels each $k$-(3,4) nucleus.
   }
    \label{fig:nd-example}
\end{figure}

\myparagraph{Example}
 An example of our algorithm for $(r,s)=(3,4)$ is shown in 
 Figure \ref{fig:nd-example}. There are 14 total triangles in 
 the example graph, namely those given by any three vertices in $\{a, b, c, d, e\}$, 
 and the additional triangles $abf$, $bef$, $aef$, and $cdg$. At 
 the start of the algorithm, $cdg$ is incident to no 4-cliques, 
 while $abf$, $bef$, and $aef$ are each incident to one 4-clique. 
 Also, $abe$ is incident to three 4-cliques, and the rest of the triangles
 are incident to two 4-cliques. Thus, only $cdg$ is peeled in the first round, 
 and has a (3, 4)-clique-core number of 0. Then, $abf$, $bef$, and $aef$ are
 peeled simultaneously in the second round, each with a 4-clique count of one, 
 which is also their (3, 4)-clique-core number. Peeling these triangles updates the 
 4-clique count of $abe$ to two, and in the third round, all remaining
 triangles have the same 4-clique count (and form the 2-$(3,4)$ nucleus) 
 and are peeled simultaneously, completing the algorithm.

\begin{figure}[t]
\centering
     \includegraphics[width=\columnwidth]{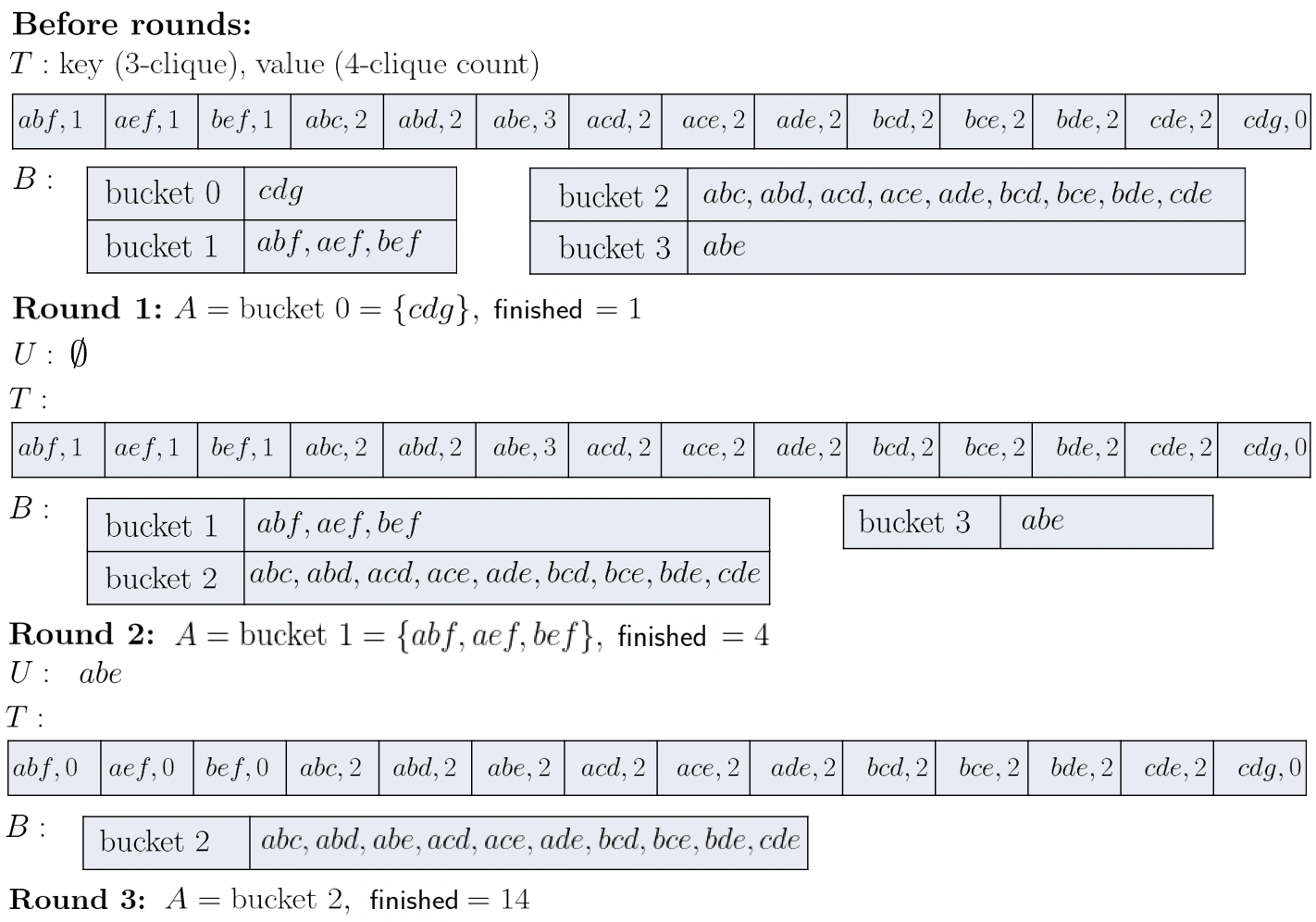}
   \caption{\color{change} An example of the data structures in \ournd{} during each round of $(3, 4)$ nucleus decomposition, on the graph in Figure \ref{fig:nd-example}.}
    \label{fig:walkthrough-nd-example}
\end{figure}

\myparagraph{Our Algorithm}
We now provide a more detailed description of the algorithm. \algprefix~\ref{alg:nd} 
presents the pseudocode for \ournd{}. We refer to Figure \ref{fig:walkthrough-nd-example},
which shows the state of each data structure after each round 
of \ournd{}, for an example of {\color{change} $(3, 4)$} nucleus decomposition on the graph in Figure~\ref{fig:nd-example}.
\ournd{}  
first directs the graph $G$ such that every vertex has out-degree 
$O(\alpha)$ (Line~\ref{line:orient}), using an efficient low out-degree 
orientation algorithm by Shi \textit{et al.} \cite{shi2020parallel}. 
Then, it initializes a parallel hash table $T$ to store $s$-clique 
counts, keyed by $r$-clique counts, and calls the recursive subroutine 
\ourcountrec{} to count and store the number of $s$-cliques incident on 
each $r$-clique (Lines~\ref{line:initialize-t}--\ref{line:rec-count}). 
It uses \algname{count-func} (Lines \ref{line:count-f-begin}--\ref{line:count-f-end})
to atomically increment the count for each $r$-clique found in each discovered $s$-clique. As
shown in Figure \ref{fig:walkthrough-nd-example}, before any rounds of peeling, $T$ contains 
the {\color{change}4-clique count incident to each triangle.}
The algorithm also initializes a parallel bucketing structure $B$ that stores 
sets of $r$-cliques with the same $s$-clique counts 
(Line~\ref{line:bucketinit}).
{\color{change}
We have four initial buckets in Figure \ref{fig:walkthrough-nd-example}. The first bucket contains $cdg$, which is incident to no $4$-cliques, and the second bucket contains $abf$, $aef$, and $bef$, which are incident to exactly one 4-clique. The third bucket contains the triangles incident to exactly two 4-cliques, $abc$, $abd$, $acd$, $ace$, $ade$, $bcd$, $bce$, $bde$, and $cde$. The final bucket contains $abe$, which is incident to exactly three 4-cliques.
}

\begin{algorithm}[!t]
  \footnotesize
 \begin{algorithmic}[1]
 \State Initialize $r$, $s$ \Comment{$r$ and $s$ for $(r,s)$ nucleus decomposition}
 
 \Procedure {Count-Func}{$T$, $S$}\label{line:count-f-begin} \footnote{\label{note1}\color{change} When \algname{Count-Func} and \algname{Update-Func} are invoked on lines 17 and 22, all arguments except the last argument $S$ are bound to each function. This is because these functions are then called in \algname{rec-list-cliques}, where they take as input an $s$-clique $S$, which is precisely the last argument.} 
  \ParFor{all size $r$ subsets $R \subset S$}
 \State Atomically add $1$ to the $s$-clique count $T[R]$ \label{line:sort-r-clique}
  \EndParFor \label{line:count-f-end}
  \EndProcedure
 
  \smallskip
 \Procedure {Update-Func}{$U$, $A$, $T$, $S$}  \textsuperscript{\ref{note1}}
 \State Let $U' \leftarrow \{ R \subset S \mid |R| = r$ and $R \not\in A\}$  \label{line:update-f-begin}
 \State Let $a$ be the \# of size $r$ subsets $R \subset S$ such that $R \in A$
 \If{any $R \in U'$ has been previously peeled} \label{line:previously-peeled-check}
 \State \Return
 \EndIf
 \ParFor{$R$ in $U'$}
 \State Atomically subtract $1 / a$ from the $s$-clique count $T[R]$\label{line:update-a-count}
 \State Add $R$ to $U$\label{line:update-u}
 \EndParFor \label{line:update-f-end}
 \EndProcedure
 
  \smallskip
\Procedure{Update}{$G=(V,E), DG, A, T$} \label{line:update}
 \State Initialize $U$ to be a parallel hash table to store $r$-cliques with updated $s$-clique counts after peeling $A$ \label{line:initialize-u}
\ParFor{$R$ in $A$}
  \State $I\leftarrow$ \algname{intersect}($N_G(v)$ for $v \in R$) \Comment{Intersect the undirected neighbors of $v \in R$} \label{line:intersect-undirected-r}
  \State \ourcountrec{}($DG$, $I$, $s-r$, $R$, \algname{update-func}($U$, $A$, $T$)) \label{line:peel-count-cliques}
\EndParFor
\State \Return $U$\label{line:update-end}
\EndProcedure

 \smallskip
\Procedure {\ournd{}}{$G = (V,E)$, \algname{Orient}}
\State $DG \leftarrow$ \algname{Orient}($G$) \Comment{Apply a user-specified orientation algorithm} \label{line:orient}
\State Initialize $T$ to be a parallel hash table with $r$-cliques as keys, and $s$-clique counts as values \label{line:initialize-t}
\State \ourcountrec{}($DG$, $V$, $s$, $\emptyset$, \algname{count-func}($T$)) \Comment{Count $s$-cliques} \label{line:rec-count}
\State Let $B$ be a bucketing structure mapping each $r$-clique to a bucket based on \# of $s$-cliques \label{line:bucketinit}
\State $\mathsf{finished} \leftarrow 0$\label{line:set-finished-zero}
\While{$\mathsf{finished} < |T|$} 
\State $A \leftarrow$ $r$-cliques in the next bucket in $B$ (to be peeled)\label{line:extractbucket}
\State $\mathsf{finished} \leftarrow \mathsf{finished} + |A|$\label{line:updatefinished}
\State $U \leftarrow$ \algname{Update}$(G, DG, A, T)$ \Comment{Update \# of $s$-cliques and return $r$-cliques with changed $s$-clique counts}\label{line:updatecliques}
\State Update the buckets of $r$-cliques in $U$, peeling $A$ \label{line:updatebuckets}
\EndWhile
\State \Return $B$
  \EndProcedure
 \end{algorithmic}
\caption{Parallel $(r,s)$ nucleus decomposition algorithm}
 \label{alg:nd}
\end{algorithm}

While not all of the $r$-cliques have been 
peeled, the algorithm repeatedly obtains the $r$-cliques incident upon 
the lowest number of induced $s$-cliques
(Line~\ref{line:extractbucket}), updates the count of the number of 
peeled $r$-cliques (Line~\ref{line:updatefinished}), and updates the 
$s$-clique counts of $r$-cliques that participate in $s$-cliques with 
peeled $r$-cliques (Line~\ref{line:updatecliques}). {\color{change}In Figure 
\ref{fig:walkthrough-nd-example}, we see that in the first round, the
bucket with the least $s$-clique count is bucket 0, and $\mathsf{finished}$ is 
updated to the size of the bucket, or one triangle. 
No 4-cliques are involved, so no updates are made to $T$, and $U$ remains empty.
}
\ournd{} then
updates the buckets for $r$-cliques with changed $s$-clique counts 
(Line~\ref{line:updatebuckets}), and repeats until all $r$-cliques have 
been peeled. At the end, the algorithm returns the bucketing structure, which 
maintains the $(r, s)$-core number of each $r$-clique. 
{\color{change}In the second round in
Figure \ref{fig:walkthrough-nd-example}, the bucket with the least 4-clique count is bucket 1, containing $abf$, $aef$, and $bef$. We add 3 to $\mathsf{finished}$, making its value 4. The sole 4-clique incident to these triangles is $abef$, so when these triangles are peeled, there is one fewer 4-clique incident on $abe$. Thus, the 4-clique count stored on $abe$ in $T$ is decremented by one, and $abe$ is returned in the set $U$ as the only triangle with a changed 4-clique count. Note that the (3, 4)-clique core number of $abf$, $aef$, and $bef$ is thus 1, which is implicitly maintained upon their removal from $B$. The bucket of $abe$ is updated to $2$, since $abe$ now participates in only two 4-cliques.
Then, in the final round, the rest of the triangles are all in the peeled bucket, bucket 2, and the $\mathsf{finished}$ count is updated to 14, or the total number of triangles. The implicitly maintained (3, 4)-clique-core number of these triangles is 2, and since no triangles remain unpeeled, there are no further updates to the data structure and the algorithm returns.
}

The subroutine \algname{Update} 
(Lines~\ref{line:update}--\ref{line:update-end}) is the main subroutine of \ournd{}, 
used on Line~\ref{line:updatecliques}. It takes as additional input the directed 
graph $DG$, a set $A$ of $r$-cliques to peel, and the parallel hash table $T$ 
that stores current $s$-clique counts, and updates incident $s$-clique counts 
affected by peeling the $r$-cliques in $A$. It also returns the set of $r$-cliques
that have not yet been peeled with their decremented $s$-clique counts. 
\algname{Update} first initializes a parallel hash table $U$ to store the set
of $r$-cliques with changed $s$-clique counts (Line~\ref{line:initialize-u}).\footnote{Note that it is inefficient to initialize $U$ here in the \algname{Update} in every round, although we do so in the pseudocode for simplicity. Instead, $U$ should be initialized following Line \ref{line:set-finished-zero}, with size equal to the total number of $r$-cliques. Then, after Line \ref{line:updatebuckets}, $U$ can be efficiently cleared for the next round by rerunning the \algname{Update} subroutine solely to clear previously modified entries in $U$.}
It then considers each $r$-clique $R \in A$ being peeled, and computes the 
intersection of the undirected neighbors of the vertices in $R$, which are stored
in set $I$ (Line~\ref{line:intersect-undirected-r}). The vertices in $I$ are 
candidate vertices to form $s$-cliques from $R$, and \algname{Update} uses $I$ 
with the subroutine \ourcountrec{} to find the remaining $s-r$ vertices to complete
the $s$-cliques incident to $R$ (Line~\ref{line:peel-count-cliques}). {\color{change}
In 
the second round of Figure \ref{fig:walkthrough-nd-example}, $abf$, $aef$, and $bef$ are being peeled. The intersection of the neighbors of $a$, $b$, and $f$ adds vertex $e$ to the discovered 4-clique.
We thus find only one 4-clique incident to $abf$, and similarly, we find the same 4-clique incident to $aef$ and $bef$. }

For each $s$-clique $S$ found, \ourcountrec{} calls \algname{update-func} (Lines~\ref{line:update-f-begin}--\ref{line:update-f-end}), which first 
checks if $S$ contains $r$-cliques that were previously peeled 
(Line~\ref{line:previously-peeled-check}). If not, \algname{update-func} 
atomically subtracts $1/a$ from each $s$-clique count for each $r$-clique in $S$, 
where $a$ is the number of size $r$ subsets in $S$ that are also being 
peeled (Line~\ref{line:update-a-count}). This is to prevent 
over-counting---if $r$-cliques that participate in the same $s$-clique 
are peeled simultaneously, then they will each subtract $1/a$ from the 
$s$-clique count, and the total subtraction will sum to $1$ for this $s$-clique. 
\algname{update-func} also adds each $r$-clique with updated counts to $U$ 
(Line~\ref{line:update-u}). {\color{change} In the second round of Figure \ref{fig:walkthrough-nd-example},
since $abf$, $aef$, and $bef$ each discover the 4-clique $abef$, which consists of $a = 3$ triangles that are simultaneously being peeled, $abf$, $aef$, and $bef$ each decrement $1/3$ from the 4-clique count of $abe$ in $T$.
 In
total, $1$ is decremented from $abe$'s 4-clique count. Then, $abe$ is
added to $U$, since its 4-clique count has been updated. 
In the third round, 
because all remaining triangles are being peeled in $A$, the 
set $U'$ defined on  Line~\ref{line:update-f-begin}, is either empty 
or contains a previously peeled triangle, by definition. Thus, 
\algname{Update} returns without performing any modifications 
to $U$ or $T$, and no buckets are updated. Afterwards,
the $\mathsf{finished}$ variable equals the total number of triangles, 
and the algorithm finishes.
}

We now discuss the theoretical efficiency our parallel nucleus decomposition algorithm.
To show that our algorithm improves upon the best existing
work bounds for the sequential nucleus decomposition algorithm, we use the following key lemma that upper bounds the sum of the minimum degrees of vertices over all
$c$-cliques.  We note that Chiba and Nishizeki~\cite{ChNi85} 
first proved this lemma for $c = 2$, and Eden \textit{et al.} first proved this lemma for larger $c$~\cite{EdRoSe20}. We provide a similar proof here.

\begin{lemma}[{{\cite{EdRoSe20}}}]\label{lem-malpha-bound}
For a graph $G$ with arboricity $\alpha$, over all $c$-cliques $C_c = \allowbreak \{v_1, \ldots, v_c\}$ in
$G$ where $c\geq 1$, $\sum_{C_c} \min_{1 \leq i \leq c}\text{deg}(v_i) = O(m \alpha^{c-1})$.
\end{lemma}

\begin{proof}
First, the $c=1$ case is easy, since $\sum_{v_i \in V} \text{deg}(v_i) = 2m$, and 
the $c=2$ case is proven by Chiba and
Nishizeki~\cite{ChNi85}. We consider $c \geq 3$ for the rest of the proof.

We begin by rewriting the sum $\sum_{C_c} \min_{1 \leq i \leq c}\text{deg}(v_i)$ to
be taken over all edges $e \in E$ instead of over all $c$-cliques $C_c$. Importantly,
we must assign each unique $c$-clique $C_c$ to a unique edge $e$ contained within 
$C_c$, to avoid double counting. We perform this assignment using 
an $O(\alpha)$-orientation of the graph, where each $c$-clique $C_c$ is an ordered 
list of vertices $v_i$ for $1 \leq i \leq c$, with the ordering given by the 
$O(\alpha)$-orientation. We assign each $C_c$ to the edge $e$ given by its first
two vertices $v_1$ and $v_2$.

Then, rewriting the sum $\sum_{C_c} \min_{1 \leq i \leq c}\text{deg}(v_i)$ to be 
taken over all edges, we obtain 
$\sum_{e \in E} \sum_{C_c \text{ assigned to }e} \min_{1 \leq i \leq c}\text{deg}(v_i)$.
For each edge $e$, the quantity inside the inner sum is upper bounded by the minimum of the degrees of the
endpoints of $e$, so we have 
$\sum_{e = (u, v) \in E} \min(\text{deg}(u), \allowbreak\text{deg}(v)) \cdot ($\# of $C_c$ assigned to $e)$.

We now claim that the number of $c$-cliques $C_c$ assigned to each edge 
$e = (u, v)$ is upper bounded by $O(\alpha^{c-2})$. This fact follows by virtue of 
the $O(\alpha)$-orientation used earlier to assign $c$-cliques $C_c$ to edges. Namely,
if we let $DG$ denote the directed graph given by the $O(\alpha)$-orientation, then
by construction for each $C_c$, vertices $v_i \in C_c$ for $i > 2$ must be 
out-neighbors of $v_1$ and $v_2$. There are at most $O(\alpha)$ out-neighbors in
the intersection of $N_{DG}(u)$ and $N_{DG}(v)$ that could potentially complete a directed
$c$-clique with the vertices $u$ and $v$. With $c-2$ additional vertices necessary to
complete a directed $c$-clique, and $O(\alpha)$ vertices to choose from, there are at most $O(\alpha^{c-2})$ such directed $c$-cliques.

Thus, our desired sum is given by $\sum_{e = (u, v) \in E} \min(\text{deg}(u), \allowbreak\text{deg}(v)) \cdot ($\# of $C_c$ assigned to $e) = \allowbreak O(\alpha^{c-2} \sum_{e = (u,v) \in E} \min(\text{deg}(u), \allowbreak\text{deg}(v)))$. By Lemma 2 of~\cite{ChNi85}, we know that $ \sum_{e=(u,v) \in E}
\min(\text{deg}(u), \text{deg}(v)) =\allowbreak O(m\alpha)$. Therefore, in total, we have $\sum_{C_c} \min_{1 \leq i \leq c}\text{deg}(v_i) = \allowbreak O(m \alpha^{c-1})$, as desired.
\end{proof}

Using this key lemma, we prove the following complexity
bounds for our parallel nucleus decomposition algorithm. 
$\rho_{(r,s)}(G)$ is defined to be the \defn{$\boldsymbol{(r,s)}$ peeling complexity} of
$G$, or the number of rounds needed to peel the graph where in each
round, all $r$-cliques with the minimum $s$-clique count are peeled.  $\rho_{(r,s)}(G) \leq O(m\alpha^{r-2})$, since at least
one $r$-clique is peeled in each round.

\begin{theorem}\label{thm:peelexact}
  \ournd{} computes the $(r,s)$ nucleus decomposition in $\BigO{m\alpha^{s-2} +
  \rho_{(r,s)}(G)\log n}$ amortized expected work and $\BigO{\rho_{(r,s)}(G) \log n + \log^2 n}$ span \whp{}, where
  $\rho_{(r,s)}(G)$ is the $(r,s)$ peeling complexity of $G$.
\end{theorem}

\begin{proof}
First, the work and span of counting the number of $s$-cliques per $r$-clique is given 
directly by Shi \emph{et al.}~\cite{shi2020parallel}, since we use this $s$-clique enumeration algorithm, \ourcountrec{},
as a subroutine. Because we hash each $s$-clique count per $r$-clique in parallel, this takes 
$O(m\alpha^{s-2})$ work and $O(\log^2 n)$ span \whp{} for a constant $s$.

We now discuss the work and span of obtaining the set of $r$-cliques with minimum $s$-clique count and updating the $s$-clique counts in our bucketing structure $B$. We make use of the fact that the total number of $c$-cliques in $G$ is bounded by $O(m \alpha^{c-2})$, which follows from the $c$-clique 
enumeration algorithm~\cite{shi2020parallel}.
The overall work of inserting $r$-cliques into $B$ is given by the number of $r$-cliques 
in $G$, or $O(m \alpha^{r-2})$. Each $r$-clique has its bucket decremented at most once
per incident $s$-clique, and because there are at most $O(m \alpha^{s-2})$ $s$-cliques,
the work of updating buckets is given by $O(m \alpha^{s-2})$. Finally, extracting the
minimum bucket can be done in $O( \log n)$ amortized expected work and $O( \log n)$ span
\whp{}, which in total gives  $O(\rho_{(r,s)}(G) \log n)$ amortized expected work, and $O(\rho_{(r,s)}(G) \log n)$ span \whp{}

Finally, it remains to discuss the work and span of obtaining updated $s$-clique counts
after peeling each set of $r$-cliques, or the work and span of the \algname{Update} subroutine. 
For each $r$-clique $R$, \algname{Update} first computes the intersection of the 
neighbors of each vertex $v \in R$, and stores them in set $I$. Notably, the total work
of intersecting the neighbors of each vertex $v \in R$ over all $r$-cliques $R$ is 
given by $O(\sum_{R} \min_{1 \leq i \leq r}\text{deg}(v_i)) = O(m \alpha^{r-1})$ \whp{}, which 
follows from Lemma \ref{lem-malpha-bound}, and the span across all intersections is $O(\log n)$ \whp{} for a constant $r$.

\algname{Update} then calls the \ourcountrec{} subroutine to complete $s$-cliques from $R$, 
taking as input the sets $I$ computed in the previous step.
Each successive recursive call to \ourcountrec{} takes a multiplicative 
$O(\alpha)$ work \whp{} due to the intersection operation, and
\ourcountrec{} requires $s-r$ recursive levels to complete each $s$-clique. This 
results in a total multiplicative factor of $O(\alpha^{s-r-1})$ work \whp{}, 
because the final recursive level (as shown in the {\color{change} $r\ell = 1$} case in \ourcountrec{})
does not involve any intersection operations, and simply iterates through the discovered
$s$-cliques. Thus, including the initial cost of setting up the call to \ourcountrec{}, 
and using the fact that the number of potential neighbors in the sets $I$ across all
$r$-cliques is $O(m\alpha^{r-1})$, the total work of \ourcountrec{} across all $r$-cliques
is given by $O(m\alpha^{r-1} \cdot \alpha^{s-r-1}) =\allowbreak O(m \alpha^{s-2})$ 
\whp{}
The span of each intersection on each recursive level is $O(\log n)$ \whp{}, and there
are $\rho_{(r,s)}(G)$ rounds by definition, which contributes $O(\rho_{(r,s)}(G)\log n)$ overall.
\end{proof}

{\color{change} \myparagraph{Discussion} \ournd{} is work-efficient with respect to the best sequential algorithm, and improves upon the best sequential algorithm that 
uses sublinear space in the number of $s$-cliques.
In more detail, the previous best sequential bounds were given by \sariyuce{} \textit{et al.}~\cite{Sariyuce2017}, in terms of the number of $c$-cliques containing each vertex $v$, or $ct_c(v)$, and the work of an arbitrary $c$-clique enumeration algorithm, or $RT_c$.
Assuming space proportional to the number of $s$-cliques and the number of $r$-cliques in the graph, they compute the $(r,s)$ nucleus decomposition in $\BigO{RT_r + RT_s} = \BigO{RT_s}$ work, and assuming space proportional to only the number of $r$-cliques in the graph, they compute the $(r,s)$ nucleus decomposition in $\BigO{RT_r + \sum_{v \in V(G)} ct_r(v) \cdot \text{deg}(v)^{s-r}}$ work.
We discuss these bounds in detail and provide a comparison to \ournd{} in the \intheapp{}.
}
Specifically, \ournd{} uses space proportional to the number of $r$-cliques
due to the space required for $T$ and $B$, and is work-efficient with respect to the 
best sequential algorithm that uses the same space, {\color{change} improving upon the corresponding bound given by \sariyuce{} \textit{et al.}~\cite{Sariyuce2017}}.
If more space is permitted, a small modification of our algorithm yields a sequential nucleus
decomposition algorithm that performs $O(m\alpha^{s-2})$ work, {\color{change} matching the corresponding bound given by \sariyuce{} \textit{et al.}~\cite{Sariyuce2017}}. The idea is to use a dense bucketing structure with space
proportional to the total number of $s$-cliques, since finding the next bucket with 
the minimum $s$-clique count involves a simple linear search rather than a heap operation. 
In the parallel setting, \ournd{} can work-efficiently find the minimum 
non-empty bucket by performing a series of steps, where at each step $i$, 
\ournd{} searches in parallel the region $[2^i, 2^{i+1}]$ for the next non-empty
bucket. This search procedure takes logarithmic span, and is work-efficient with respect to
the sequential algorithm.

\section{Practical Optimizations}
We now introduce the practical optimizations that we use for our parallel $(r,s)$ nucleus decomposition implementation.

\subsection{Number of Parallel Hash Table Levels}\label{sec:multilevel-opt}

\ournd{} uses a single parallel hash table $T$ to store the $s$-clique counts, where the keys are $r$-cliques. 
However, this storage method is infeasible in practice due to space limitations, particularly for large $r$, since $r$ vertices must be concatenated into a key for each $r$-clique. 
We observe that space can be saved by introducing more levels to our parallel hash table $T$.
For example, one option is to instead use a two-level combination of an array and a parallel hash table, 
which consists of an array of size $n$ whose elements are pointers to individual hash tables 
where the keys are $(r-1)$-cliques. The $s$-clique count for a corresponding $r$-clique 
$R = \{ v_1, \ldots, v_r\}$ (where the vertices are in sorted order) is stored by indexing into 
the $v_1^\text{th}$ element of the array, and storing the count on the key corresponding to the 
$(r-1)$-clique given by $\{ v_2, \ldots, v_r\}$ in the given hash table. Space 
savings arise because $v_1$ does not have to be repeatedly stored for each 
$(r-1)$-clique in its corresponding hash table, but rather is only stored once 
in the array.

A more general option, particularly for large $r$, is to use a multi-level parallel hash table, 
with $\ell \leq r$ levels in which each intermediate level consists of a parallel hash table 
keyed by a single vertex in the $r$-clique whose value is a pointer to a parallel hash table in the subsequent level, and the last level consists of a parallel hash table
keyed by $(r-\ell+1)$-cliques. 
Given an $r$-clique $R = \{ v_1, \ldots, v_r\}$ (where the 
vertices are in sorted order), each of the vertices in the clique in order is mapped to each level of the hash table, except for the last $(r-\ell + 1)$ vertices which are concatenated into a key for the last level of the hash table. Thus, the location in the hash table corresponding to $R$ can be found by looking up the key $v_j$ in the hash table at each level for $j < \ell$, and following the pointer to the next level. At the last level, the key is given by the $(r-\ell + 1)$-clique corresponding to $\{ v_{\ell}, \ldots, v_{r}\}$.
Again, space savings arise due to the shared vertices on each intermediate level, which 
need not be repeatedly stored in the keys on the subsequent levels of the parallel hash table,
and for $r \geq \ell > 2$, these savings may exceed those found in the previous combination
of an array and a parallel hash table.

In considering different numbers of levels, we differentiate between the 
\emph{two-level} combination of an array and a parallel hash table, and the  \emph{$\ell$-multi-level} option of nested parallel hash tables, where $\ell$ is the 
number of levels, and notably, we may have $\ell = 2$. {\color{change} Figures \ref{fig:levels-example} and \ref{fig:levels-example-2} show examples of $T$ using different numbers of levels. As shown in these examples, there are cases where increasing the number of levels does not save space, because the additional pointers required in increasing the number of levels exceeds the overlap in vertices between $r$-cliques, particularly for small $r$. 

\begin{figure}[t]
\centering
    \includegraphics[width=\columnwidth]{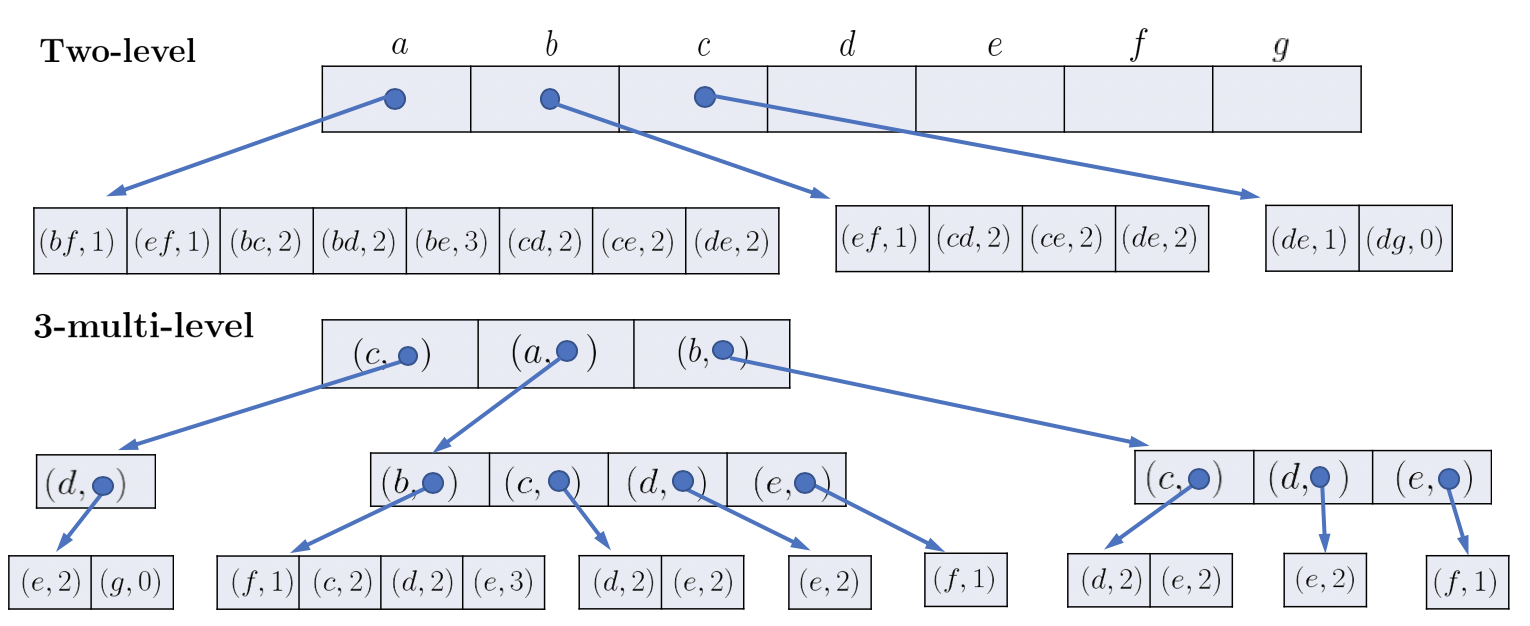}
   \caption{\color{change} An example of the initial parallel hash table $T$ for $(3, 4)$ nucleus decomposition on the graph from Figure \ref{fig:nd-example}, considering different numbers of levels. Note that Figure \ref{fig:walkthrough-nd-example} shows a one-level parallel hash table $T$. If we consider each vertex and each pointer to take a unit of memory, the one-level $T$ takes 42 units, while the two-level $T$ takes 35 units, thus saving memory. However, the 3-multi-level $T$ takes 50 units, because $r = 3$ is too small to give memory savings for this graph.
   }
    \label{fig:levels-example}
\end{figure}

\begin{figure}[t]
\centering
    \includegraphics[width=0.7\columnwidth]{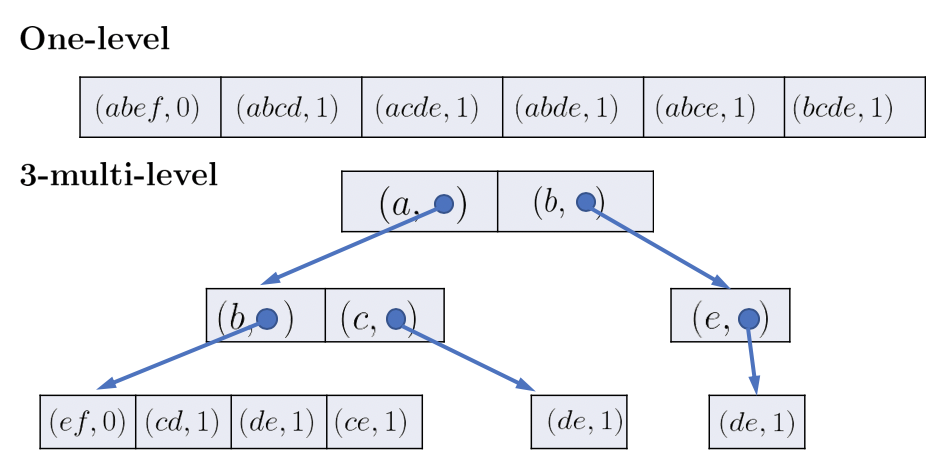}
   \caption{\color{change} An example of the initial parallel hash table $T$ for $(4, 5)$ nucleus decomposition on the graph from Figure \ref{fig:nd-example}, considering different numbers of levels. If we consider each vertex and each pointer to take a unit of memory, the one-level $T$ takes 24 units, while the 3-multi-level $T$ takes 22 units, thus saving memory. We see memory savings with more levels in $T$ compared to in Figure \ref{fig:levels-example}, because $r = 4$ is sufficiently large.
   }
    \label{fig:levels-example-2}
\end{figure}

\begin{figure}[t]
\centering
   \centering
    \includegraphics[width=\columnwidth]{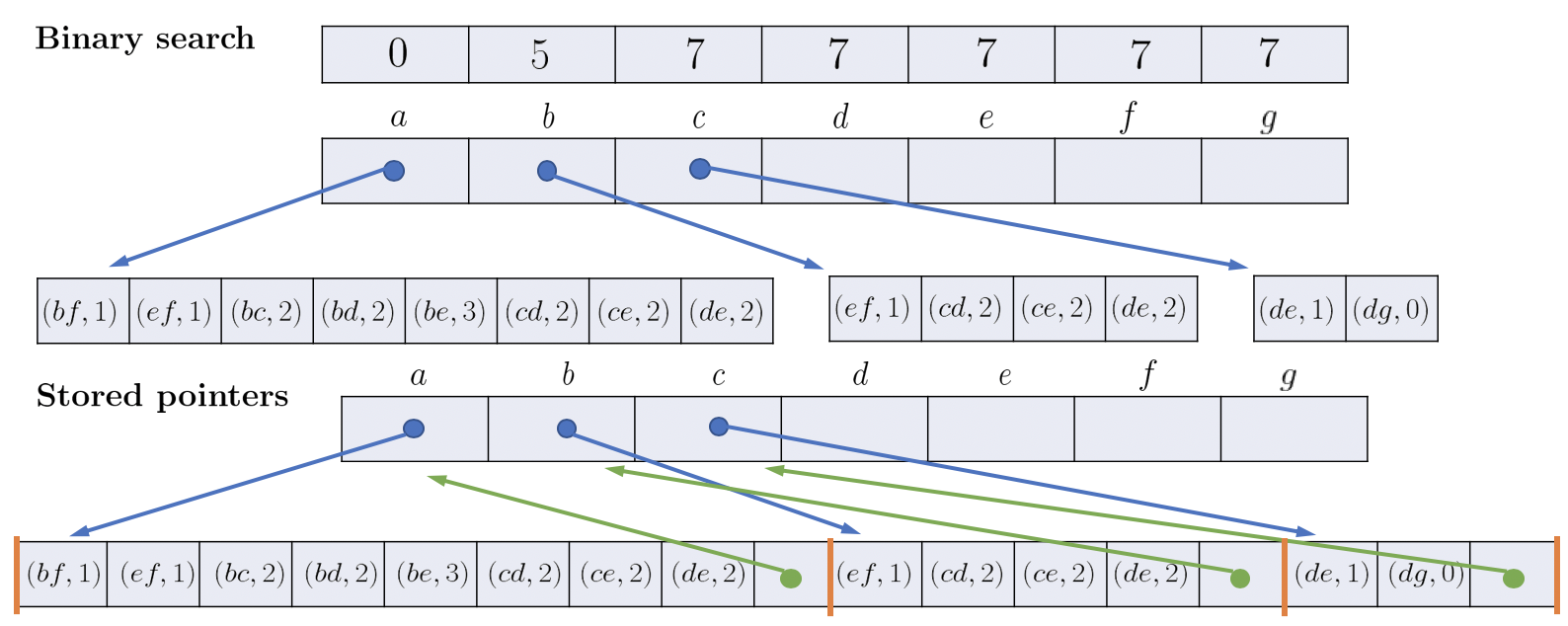}
   \caption{\color{change} Using the same graph as shown in Figure \ref{fig:nd-example} and the two-level $T$ from Figure \ref{fig:levels-example}, for $(3, 4)$ nucleus decomposition, an example of the binary search and stored pointers methods.
   }
    \label{fig:inv-map-example}
\end{figure}

\begin{figure}[t]
\centering
   \centering
    \includegraphics[width=0.6\columnwidth]{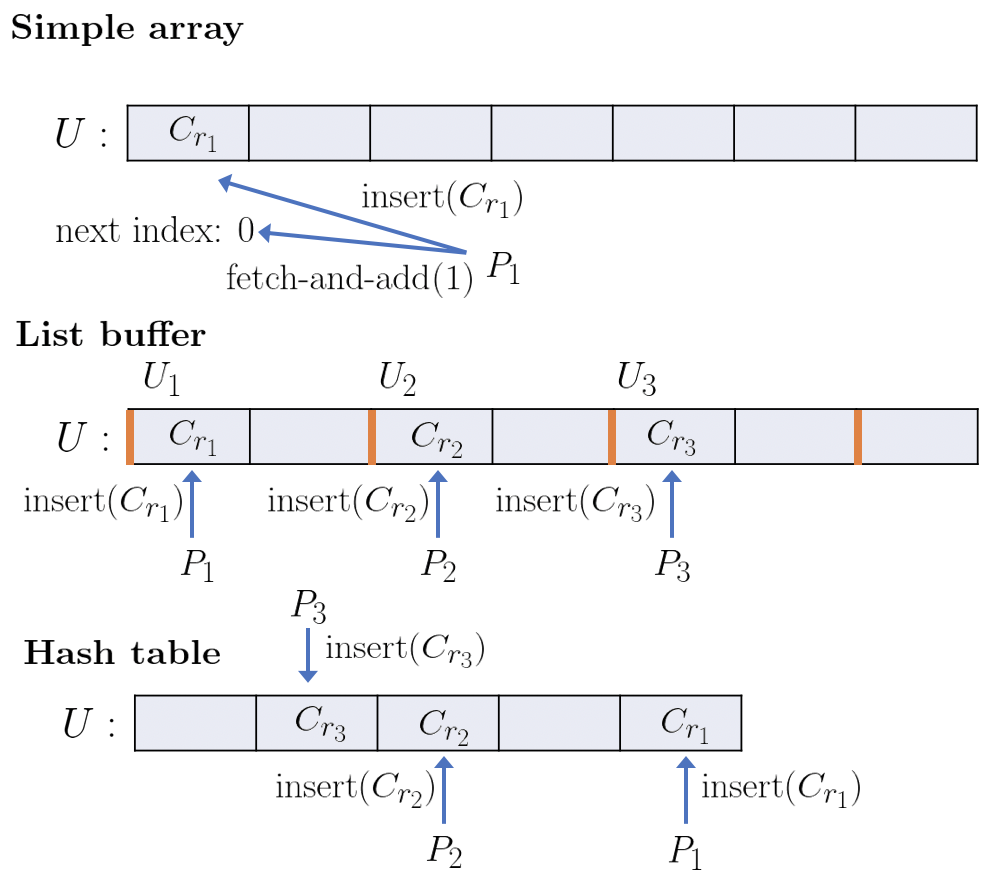}
   \caption{\color{change} An example of the simple array, list buffer, and hash table options in aggregating the set $U$ of $r$-cliques with updated $s$-clique counts. Processors $P_1$, $P_2$, and $P_3$ are storing $r$-cliques $C_{r_1}$, $C_{r_2}$, and $C_{r_3}$, respectively, in $U$.
   }
    \label{fig:agg-example}
\end{figure}

Moreover, note that there are no theoretical guarantees on whether additional levels will result in a memory reduction, since a memory reduction depends on the number of overlapping $r$-cliques and the size of the overlap in $r$-cliques in any given graph. In other words, the skew in the constituent vertices of the $r$-cliques directly impacts the possible memory reduction.
Similarly, performance improvements may arise due to better cache locality in accessing $r$-cliques that share many vertices, but this is not guaranteed, and cache misses are inevitable while traversing through different levels.

The idea of a multi-level parallel hash table is more generally applicable in scenarios where the efficient storage and access of sets with significant overlap is desired, particularly if memory usage is a bottleneck. An example use case is to efficiently store the hyperedge adjacency lists for a hypergraph. The multi-level parallel hash table can also be used for data with multiple fields or dimensions, where each level keys a different field.
}

\subsection{Contiguous Space}\label{sec:contiguous}
In using the two-level and multilevel data structures as $T$, 
a natural way to implement the last level parallel hash tables in these more complicated data 
structures is to simply allocate separate blocks of memory as needed for each last level parallel
hash table. However, this approach may be memory-inefficient and lead to poor cache locality.
An optimization for the two-level and multi-level tables is to instead first compute the space needed for all last level parallel hash 
tables and then allocate a contiguous block of memory such that the last level tables are 
placed consecutively with one another, using a parallel prefix sum to determine each of their 
indices into the contiguous block of memory. This optimization requires little overhead, and has the additional benefit of greater cache locality.
This optimization does not apply to one-level parallel hash tables, since they are by nature contiguous.

\subsection{Binary Search vs.\ Stored Pointers}\label{sec:pointers}

An important implementation detail from \ournd{} is the interfacing between the representation 
of $r$-cliques in the bucketing structure $B$ and the representation of $r$-cliques in the parallel
hash table $T$. It is impractical to use the theoretically-efficient Fibonacci heap used to
obtain our theoretical bounds, due to the complexity of Fibonacci heaps in general; in practice, 
we use the efficient parallel bucketing implementation by Dhulipala et al.~\cite{DhBlSh17}.

This bucketing structure requires each $r$-clique to be represented by an index, and to interface
between $T$ and $B$, we require a map translating each $r$-clique in $T$ to its index in $B$, and 
vice versa. {\color{change} More explicitly, for a given $r$-clique in $T$, we must be able to find the number of $s$-cliques it participates in using $B$, and symmetrically, for a given $r$-clique that we peel from $B$, we must be able to find its constituent vertices using $T$. It would be impractical in terms of space to represent the $r$-clique directly using its constituent vertices in $B$, since its constituent vertices are already stored in $T$. Thus, we seek a unique index to represent each $r$-clique by in $B$, that is easy to map to and from the $r$-clique in $T$. } If $T$ is represented using a one-level parallel hash table, then a natural index to 
use for each $r$-clique is its key in $T$, and the map is implicitly the identity map. {\color{change} For instance, considering the one-level $T$ in Figure \ref{fig:walkthrough-nd-example}, instead of storing $cdg$ in bucket 0 before any rounds begin, we would instead store 13, which is the index of $cdg$ in $T$.} However, 
if $T$ is represented using a two-level or multi-level structure, then the index representation 
and map are less natural, and require additional overheads to maintain.

One important property of the mapping from $T$ to $B$ is space-efficiency, and so it is
desirable to maintain an implicit map. Therefore, we represent each $r$-clique by the unique index 
corresponding to its position in the last level parallel hash table in $T$, which can be implicitly 
obtained by its location in memory if $T$ is represented contiguously. 
If $T$ is not represented 
contiguously, then we can store the prefix sums of the sizes of each successive level of 
parallel hash tables, and use these plus the $r$-clique's index at its last level table in $T$ to 
obtain the unique index.
{\color{change} For equivalent $T$, the index corresponding to each $r$-clique is the same regardless of whether $T$ is represented contiguously in memory or not. For instance, considering the two-level $T$ in Figure \ref{fig:levels-example}, instead of storing $abc$ in bucket 2 before any rounds begin, we would instead store 2, which is the index of $abc$ in the second level of $T$. Similarly, instead of storing $bef$ in bucket 1, we would instead store 8; this is because we take the index corresponding to $bef$ to be its index in the second level hash table corresponding to vertex $b$, plus the sizes of earlier second level hash tables, namely the size of the hash table corresponding to vertex $a$. }
Then, it is also necessary to implicitly maintain the inverse map, from an index 
in $B$ to the constituent vertices of the corresponding $r$-clique, for which we provide two methods.

\myparagraph{Binary Search Method}
One solution is to store the prefix sums of the sizes of each successive level of parallel hash tables in 
both the contiguous and non-contiguous cases, and use a serial binary search to find the table 
corresponding to the given index at each level. We also store the vertex corresponding to each 
intermediate level alongside each table, which can be easily accessed after the binary 
search. The key at the last level table can then be translated to its constituent $r-\ell$ vertices. 
This does not require the last level parallel hash tables to be stored contiguously in memory. Figure \ref{fig:inv-map-example} shows an example of this method in the non-contiguous case, where the top array is the prefix sum of the sizes of the second level hash tables. Each entry in the prefix sum corresponds to an entry in the first level, which contains the first vertex of the triangle and points to the second level, which contains the other two vertices.

\myparagraph{Stored Pointer Method}
However, while using binary searches is a natural solution, they may be computationally inefficient, both in terms of work and cache misses,
considering the number of such translations that are needed throughout \ournd{}. Instead, we
consider a more sophisticated solution which, assuming contiguous memory, is to place barriers
between parallel hash tables on each level that contain up-pointers to prior levels, and fill 
empty cells of the parallel hash tables with the same up-pointers. Then, given an index to the
last level, which translates directly into the memory location of the cell containing the last
level $r$-clique key due to the contiguous memory, we can perform linear searches
to the right until it reaches an empty cell, which directly gives an up-pointer to the prior level
(and thus, also the vertex corresponding to the prior level).
We differentiate between empty and non-empty cells by reserving the top bit of each key to indicate whether the cell is empty or non-empty. The benefit of this method is 
that with a good hashing scheme, the linear search for an empty cell will be faster and more 
cache-friendly compared to a full binary search.
Figure \ref{fig:inv-map-example} also shows an example of this method using contiguous space, in which the barriers placed to the right of each hash table in the second level point back to the parent entry in the first level, allowing an index to traverse up these pointers to obtain the corresponding first vertex of the triangle. Note that the bold orange lines on the second level mark the boundaries between parallel hash tables corresponding to different vertices from the first level.

With both the contiguous space and stored pointers optimizations applied to two-level and multi-level parallel hash tables $T$, 
we show in Section \ref{sec:eval-tuning} that we achieve up to a 1.32x speedup of using a two-level $T$, and up to a 1.46x speedup of using an $\ell$-multi-level $T$ for $\ell > 2$, over one level.

\subsection{Graph Relabeling}\label{sec:relabeling}
We sort the vertices in each $r$-clique prior to translating
it into a unique key for use in the parallel
hash table $T$. However, the lexicographic ordering of vertices in the 
input graph may not be representative of the access patterns used in
finding $r$-cliques. 
In particular, because we use directed edges based on an
$O(\alpha)$-orientation to count $r$-cliques and $s$-cliques, an 
optimization that could save work and improve cache locality in 
accessing $T$ is to relabel vertices based on the 
$O(\alpha)$-orientation, so that the sorted order is representative 
of the order in which vertices are discovered in the
\ourcountrec{} subroutine. 
The first benefit of this optimization is that there is no need to 
re-sort $r$-cliques based on the orientation, which is implicitly 
performed on Line \ref{line:sort-r-clique} of Algorithm \ref{alg:nd}, as after relabeling, the vertices in a clique will be added in increasing order. 
A second benefit is that $r$-cliques discovered in the same 
recursive hierarchy will be located closer together in our parallel 
hash table, potentially leading to better cache locality when accessing
their counts in $T$.
In our evaluation in Section \ref{sec:eval-tuning}, we see up to a 1.29x speedup using graph relabeling.

\subsection{Obtaining the Set of Updated $r$-cliques}\label{sec:setU}
A key component of the bucketing structure in \ournd{} is the computation of the set $U$ 
of $r$-cliques with updated $s$-clique counts, after peeling a set of $r$-cliques. Using a 
parallel hash table with size equal to the number of $r$-cliques is slow in practice due to 
the need to iterate through the entire hash table to retrieve the
$r$-cliques and to clear the hash table. We present here three options for 
computing $U$. {\color{change} Figure \ref{fig:agg-example} shows an example of each of these.}

\myparagraph{Simple Array}
The first option is to represent $U$ as a simple array to hold $r$-cliques, along with a variable
that maintains the next open slot in the array. We use a fetch-and-add to update 
the $s$-clique count of a discovered $r$-clique in the \algname{Update-Func}  subroutine, and if in 
the current round this is the first such modification of the $r$-clique's count, 
we use a fetch-and-add  to reserve a slot in $U$, and store the $r$-clique in 
the reserved slot in  $U$. This method introduces contention due to 
the requirement of all updated 
$r$-cliques to perform a fetch-and-add with a single variable to
reserve a slot in $U$. {\color{change} As shown in Figure \ref{fig:agg-example}, processor $P_1$ successfully updates the index variable and inserts its $r$-clique $C_{r_1}$ into the first index in $U$, but the other processors must wait until their fetch-and-add operations succeed.} However, the $r$-cliques are compactly 
stored in $U$ and there is no need to clear elements in $U$, 
which results in time savings.

\myparagraph{List Buffer} The second option improves upon the contention incurred by the first option, offering better 
performance. We use a data structure that we call a \emph{list buffer}, which consists of an 
array $U$ that holds $r$-cliques, and $P$ variables $\{i_1, \ldots, i_P\}$ that maintain the next 
open slots in the array, where $P$ is the number of threads. Each of the $P$ variables is 
exclusively assigned to one of the $P$ threads. The data structure also uses a constant buffer 
size, and each thread is initially assigned a contiguous block of $U$ of size equal to this 
constant buffer size. When a thread $j$ is the first to modify an $r$-clique's count in a given
round, the thread updates its corresponding $i_j$, and uses the reserved open slot in $U$ to 
store the $r$-clique. If a thread runs out of space in its assigned block in $U$, it uses a 
fetch-and-add operation to reserve the next available block in $U$ of size equal to the 
constant buffer size. We filter $U$ of all unused slots, 
prior 
to returning $U$ to the bucketing data structure. The reduced contention is due 
to the exclusive $i_j$ that each thread $j$ can update without contention. Threads may still 
contend on  reserving new blocks of space in $U$, but the contention is minimal with a large 
enough buffer size. In reusing the list buffer data structure in later
rounds, there is no need to clear $U$, as it is sufficient to reset the $i_j$.
{\color{change} In Figure \ref{fig:agg-example}, the buffer size is 2, and the slots in $U$ assigned to each processor $P_i$ are labeled by $U_i$. Each $P_i$ can store its corresponding $r$-clique $C_{r_i}$ in parallel, since there is no contention within their buffer as long as the buffer is not full.}

\myparagraph{Hash Table}
The last option reduces potential contention even further by removing the necessity to
reserve space, but at the cost of additional work needed to clear $U$ between rounds. 
This option uses a parallel hash table as $U$, but dynamically determines the amount of
space required on each round based on the number of $r$-cliques
peeled, and as such, the maximum possible number of $r$-cliques with updated $s$-clique 
counts. Thus, in rounds with fewer $r$-cliques peeled, less space is reserved for $U$ for
that round, and as a result, less work is required to clear the space to reuse in 
the next round. {\color{change} Figure \ref{fig:agg-example} shows each processor $P_i$ storing its corresponding $r$-clique $C_{r_i}$ into the parallel hash table. However, the entirety of $U$ must be cleared in order to reuse $U$ between rounds.}

In our evaluation in Section \ref{sec:eval-tuning}, we achieve up to a 3.98x speedup using a list buffer and up to a 4.12x speedup using a parallel hash table, both over a simple array.

\subsection{Graph Contraction} 
In the special case of $(2, 3)$ nucleus (truss) decomposition, we introduce an optimization that filters out peeled edges when a significant number of edges have been peeled over successive rounds. 
When many edges have been peeled, reducing the overall size of the graph and contracting adjacency lists may save work in future rounds, in that work does not need to be spent iterating over previously peeled edges, allowing for performance improvements.

We perform this contraction when the number of peeled edges since the previous contraction is at least twice the number of vertices in the graph, and we only contract adjacency lists of vertices that have
lost at least a quarter of their neighbors since the previous contraction. {\color{change} We chose these boundaries for contraction heuristically based on performance on real-world graphs; these generally}
allow for contraction to be performed only in instances in which it could meaningfully impact the 
performance of future operations. {\color{change} Specifically, we found that the overhead of each contraction operation is balanced by reduced work in future iterations only when vertices can significantly reduce the size of their adjacency lists, and iterating through vertices to check for this condition is only worthwhile when sufficiently many edges have been peeled.} Note that the $k$-truss decomposition implementation 
by Che \emph{et al.}~\cite{ChLaSuWaLu20} also periodically contracts the graph.
{\color{change} In terms of implementation details, the contraction operations are performed in parallel for each vertex. If a vertex meets the heuristic criteria, its adjacency list is filtered into a new adjacency list sans the peeled edges, using the parallel filter primitive; this new adjacency list replaces the previous adjacency list.}
This optimization does not extend to higher $(r, s)$ nucleus decomposition, because if an $r$-clique has been peeled for $r > 2$, its edges may still be involved in other unpeeled $r$-cliques.
Thus, there is no natural way to contract out an $r$-clique from a graph for $r > 2$ to allow for work savings in future rounds.
In Section \ref{sec:eval-tuning}, we see up to a 1.08x speedup using graph contraction in $(2, 3)$ nucleus decomposition.
\section{Evaluation}

\subsection{Environment and Graph Inputs}
We run experiments on a Google Cloud Platform instance of a 30-core machine with two-way hyper-threading, with 3.8 GHz Intel Xeon Scalable (Cascade Lake) processors and 240 GiB of main memory. We compile with g++ (version 7.4.0) and use the \texttt{-O3} flag. We use the work-stealing scheduler \algname{ParlayLib} by Blelloch \textit{et al.}~\cite{BlAnDh20}. We terminate any experiment that takes over 6 hours.
We test our algorithms on real-world graphs from the Stanford Network Analysis Project (SNAP)~\cite{SNAP}, shown in 
Figure~\ref{fig:core-rho} with $\rho_{(r,s)}$ and the maximum $(r, s)$-core numbers for $r < s \leq 7$. We also use synthetic rMAT graphs \cite{ChakrabartiZF04}, with $a = 0.5$, $b =c=0.1$, and $d = 0.3$.

We compare to {\color{change}\sariyuce \textit{et al.}'s state-of-the-art parallel~\cite{sariyuce2017parallel} and serial~\cite{Sariyuce2017} nucleus decomposition implementations}, which address $(2, 3)$ and $(3, 4)$ nucleus decomposition.
For the special case of $(2, 3)$ nucleus decomposition, we compare 
to Che \textit{et al.}'s~\cite{ChLaSuWaLu20} highly optimized state-of-the-art parallel \algname{pkt-opt-cpu}, {\color{change} and implementations by Kabir and Madduri's \algname{pkt}~\cite{kabir2017parallel}, Smith \textit{et al.}'s \algname{msp}~\cite{smith2017truss}, and Blanco \textit{et al.}~\cite{BlLoKi19}, which represent the top-performing implementations from the MIT GraphChallenge~\cite{GraphChallenge}.} We run all of these using the same environment as our experiments on \ournd{}. 

\subsection{Tuning Optimizations}\label{sec:eval-tuning}
There are six total optimizations that we implement in \ournd{}: different numbers of levels in our parallel hash table $T$ (\emph{numbers of levels}), the use of contiguous space (\emph{contiguous space}), a binary search versus stored pointers to perform the mapping of indices representing $r$-cliques to the constituent vertices (\emph{inverse index map}), graph relabeling (\emph{graph relabeling}), a simple array method versus a list buffer versus a parallel hash table for maintaining the set $U$ of $r$-cliques with changed $s$-clique counts per round (\emph{update aggregation}), and graph contraction specifically for $(2, 3)$ nucleus decomposition (\emph{graph contraction}).

\begin{figure}[t]
\begin{subfigure}{.38\columnwidth}
\small
\centering
\setlength{\tabcolsep}{2pt}
\scalebox{0.9}{
\begin{tabular}{lll}
\toprule
& $n$         & $m$                  \\ \midrule
\textbf{amazon} &334,863   & 925,872  \\ \hline
\textbf{dblp} &317,080   & 1,049,866  \\ \hline
\textbf{youtube} & 1,134,890  & 2,987,624  \\ \hline
\textbf{skitter} &1,696,415 & 11,095,298  \\ \hline
\textbf{livejournal}.&3,997,962 & 34,681,189  \\ \hline
\textbf{orkut} & 3,072,441 & 117,185,083 \\ \hline
\textbf{friendster} & 65,608,366 & $1.806\times 10^9$  \\ \hline
\end{tabular}
}
\end{subfigure}%
      \hfill
\begin{subfigure}{.46\columnwidth}
   \centering
   \includegraphics[width=\columnwidth, page=8]{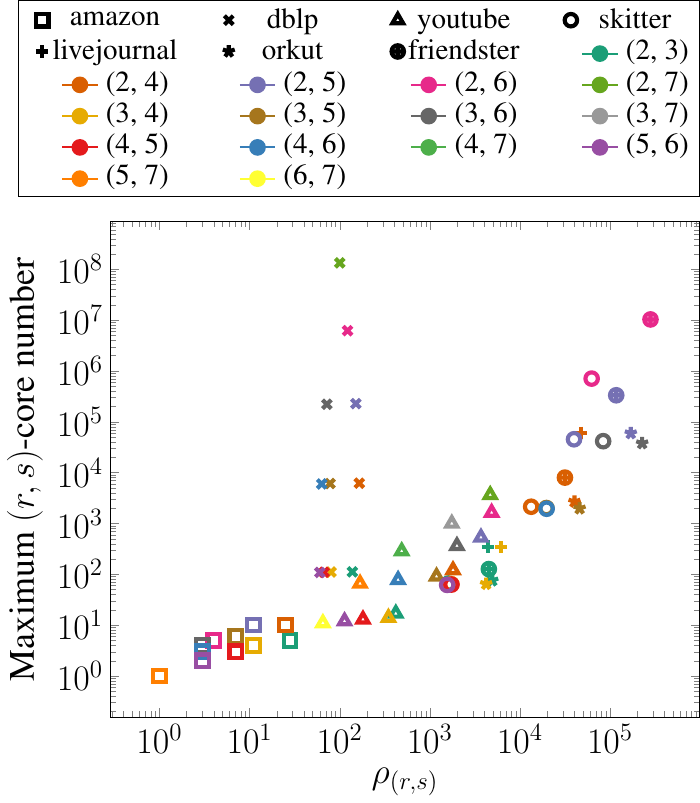}
    \end{subfigure}%
    \caption{Sizes of our input graphs, all of which are from~\cite{SNAP}, on the left, and the number of rounds required $\rho_{(r,s)}$ and the maximum $(r, s)$-core numbers for $r < s \leq 7$ for each graph on the right.
   }
    \label{fig:core-rho}
\end{figure}

\begin{figure*}[t]
    \centering
    \begin{subfigure}{.57\textwidth}
    \centering
    \includegraphics[width=\textwidth, page=3]{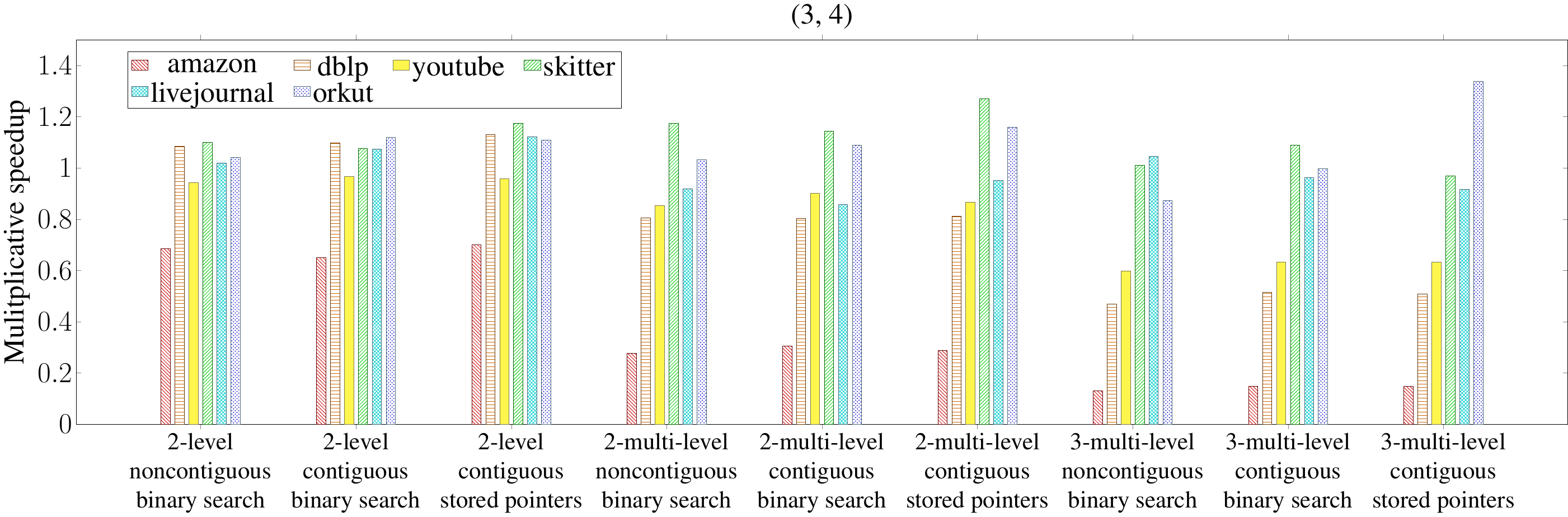}
    \end{subfigure}%
      \hfill
\begin{subfigure}{.39\textwidth}
   \centering
   \includegraphics[width=\textwidth, page=10]{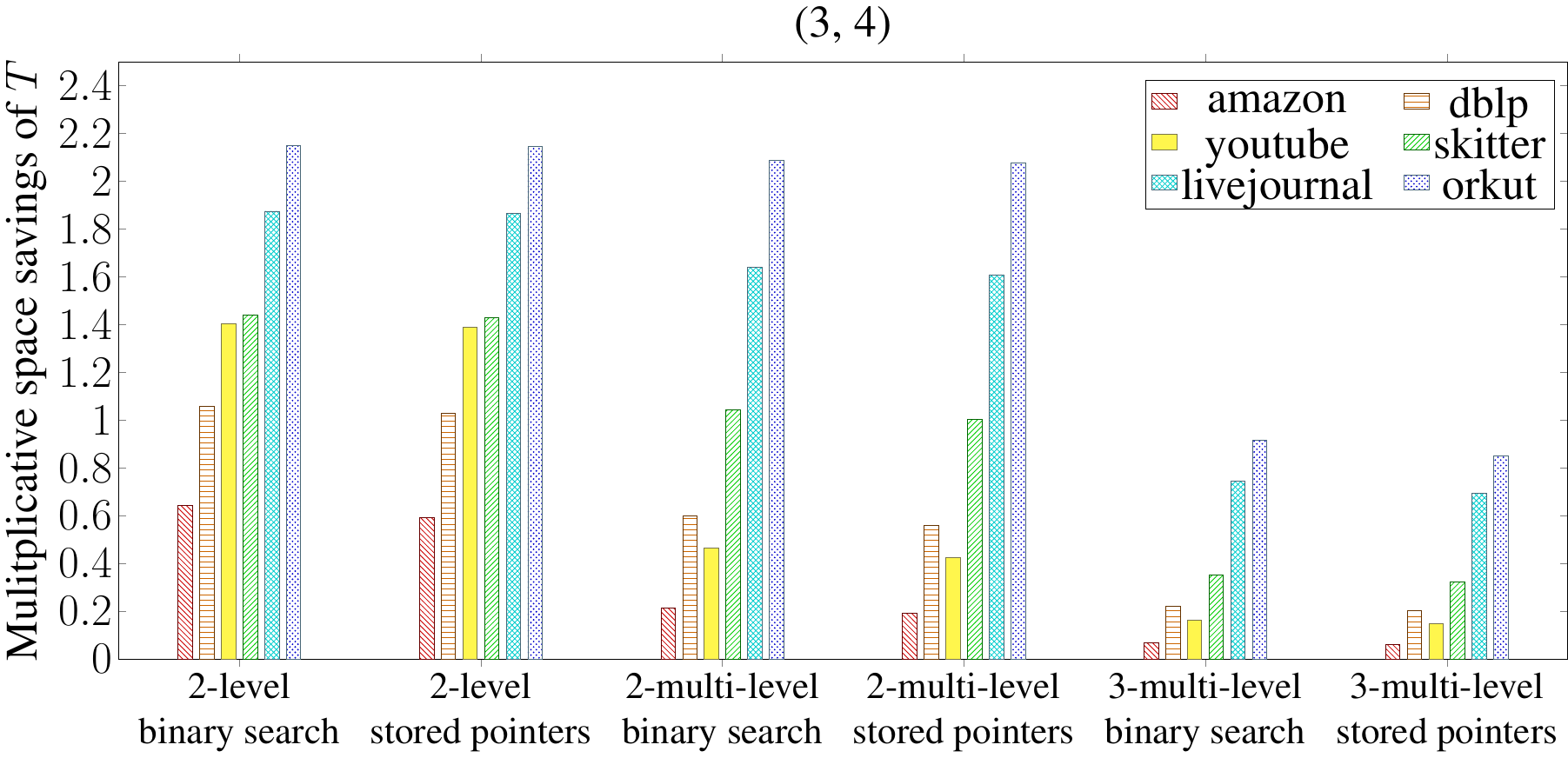}
   \end{subfigure}
   \caption{On the left, multiplicative speedups of different combinations of optimizations on $T$ in \ournd{}, over an unoptimized setting of \ournd{}, for $(3, 4)$ nucleus decomposition. On the right, multiplicative space savings for $T$ of different combinations of optimizations on $T$ in \ournd{}, over the unoptimized setting, for $(3, 4)$ nucleus decomposition. Note that the space usage between the non-contiguous option and the contiguous option is equal. Friendster is omitted because \ournd{} runs out of memory for this graph.}
    \label{fig:nd-tune-lvl}
\end{figure*}

\begin{figure*}[t]
    \centering
  \includegraphics[width=0.8\textwidth, page=1]{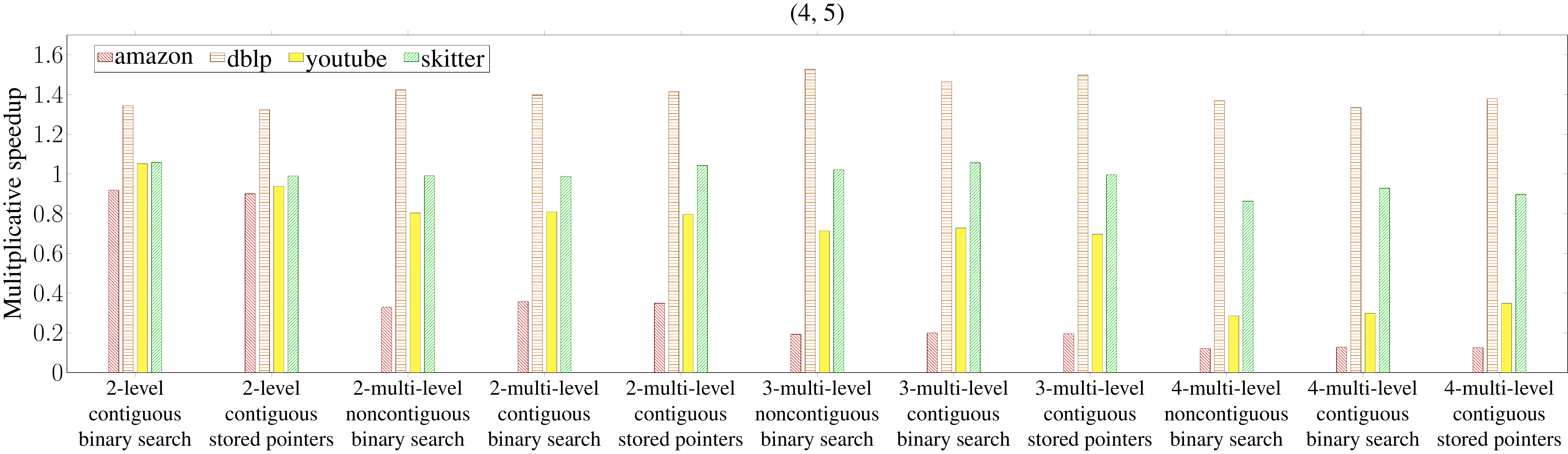}
   \caption{Multiplicative speedups of different combinations of optimizations on $T$ in \ournd{}, over an unoptimized setting of \ournd{}, for $(4, 5)$ nucleus decomposition. Livejournal, orkut, and friendster are omitted, because \ournd{} runs out of memory for these instances.   }
    \label{fig:nd-tune-lvl-45}
\end{figure*}

\begin{figure*}[t]
   \centering
   \includegraphics[width=0.7\textwidth, page=2]{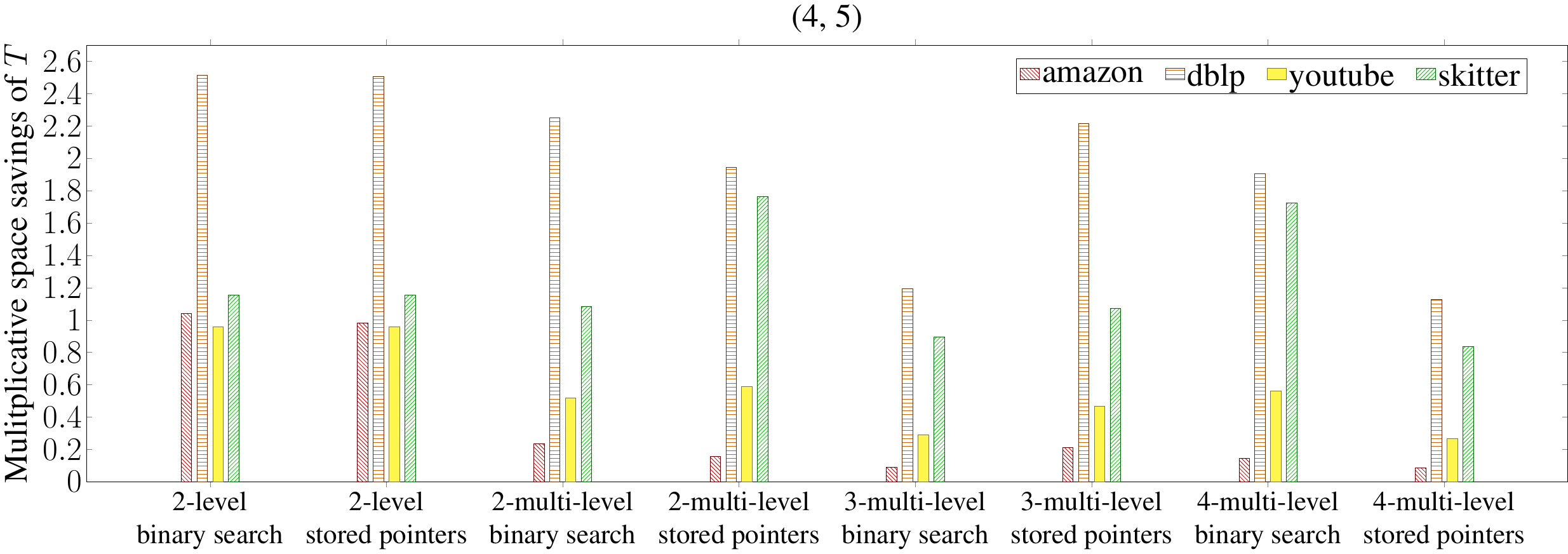}
   \caption{Multiplicative space savings for $T$ of different combinations of optimizations on $T$ in \ournd{}, over the unoptimized setting, for $(4, 5)$ nucleus decomposition. Note that the space usage between the non-contiguous option and the contiguous option is equal. Also, livejournal, orkut, and friendster are omitted, because \ournd{} runs out of memory for these instances.   }
    \label{fig:nd-tune-space-45}
\end{figure*}

\begin{figure*}[t]
    \centering
   \begin{subfigure}{.36\textwidth}
   \centering
   \includegraphics[width=\columnwidth, page=5]{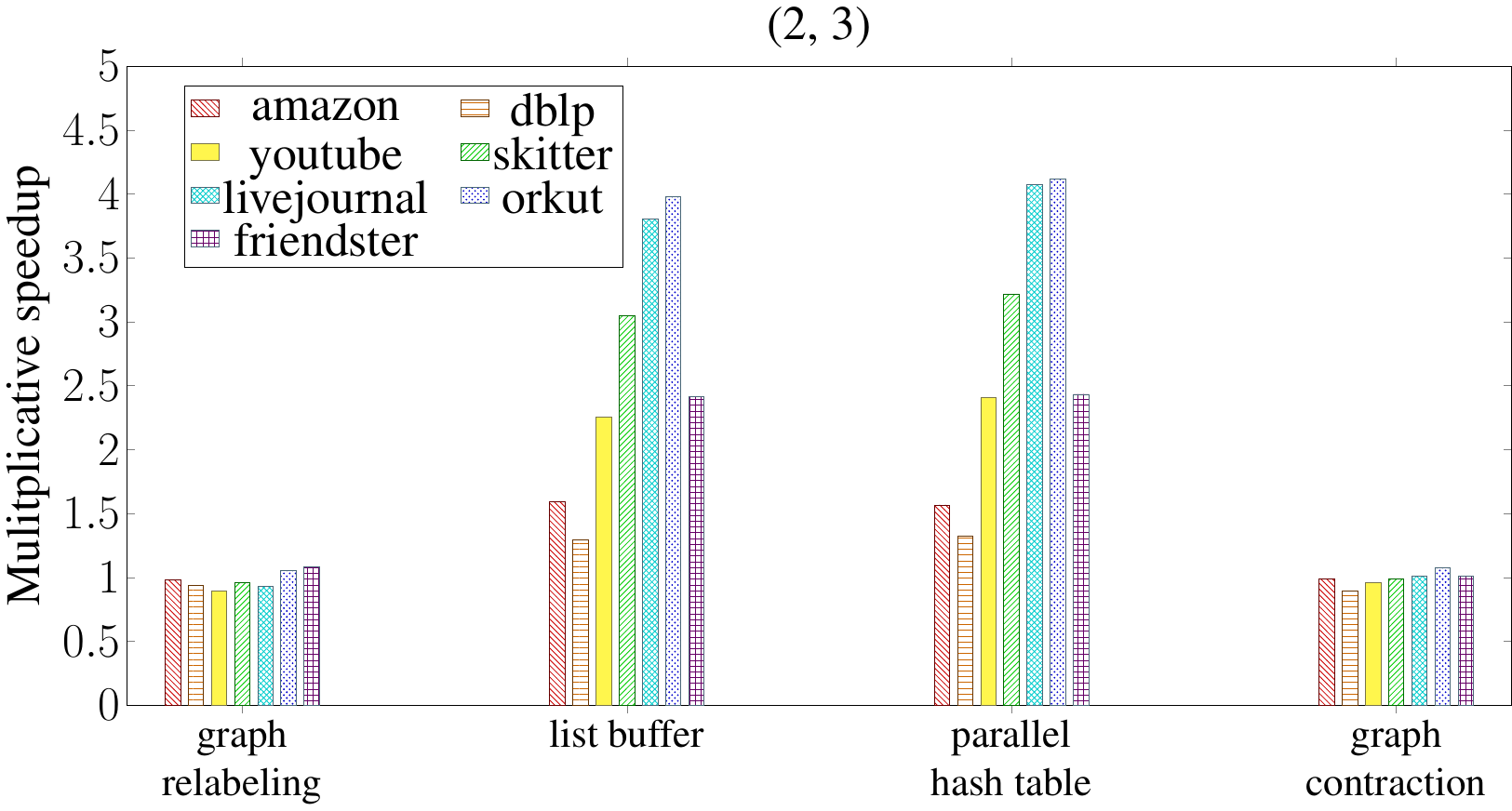}
   \end{subfigure}%
   \hfill
   \begin{subfigure}{.28\textwidth}
   \centering
   \includegraphics[width=\columnwidth, page=7]{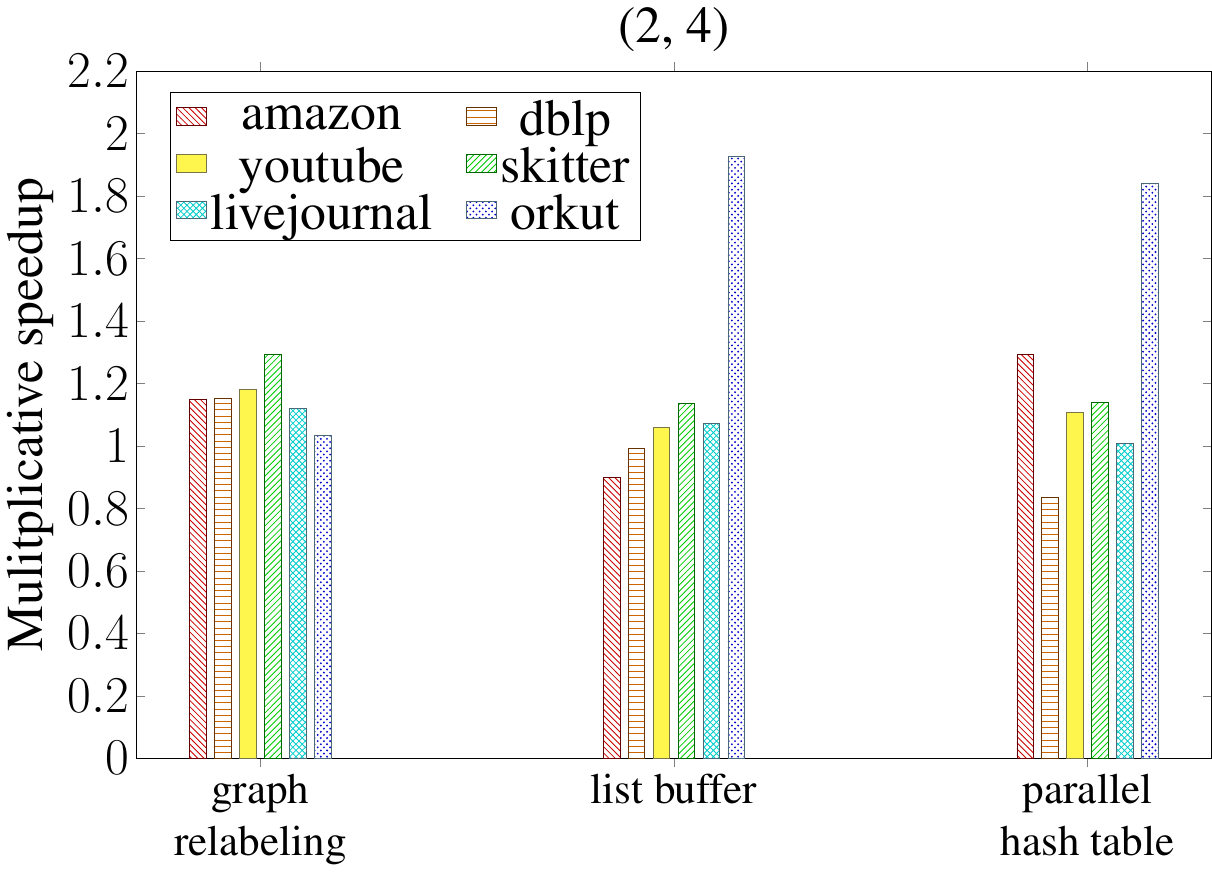}
      \end{subfigure}%
   \hfill
    \begin{subfigure}{.28\textwidth}
   \centering
   \includegraphics[width=\columnwidth, page=9]{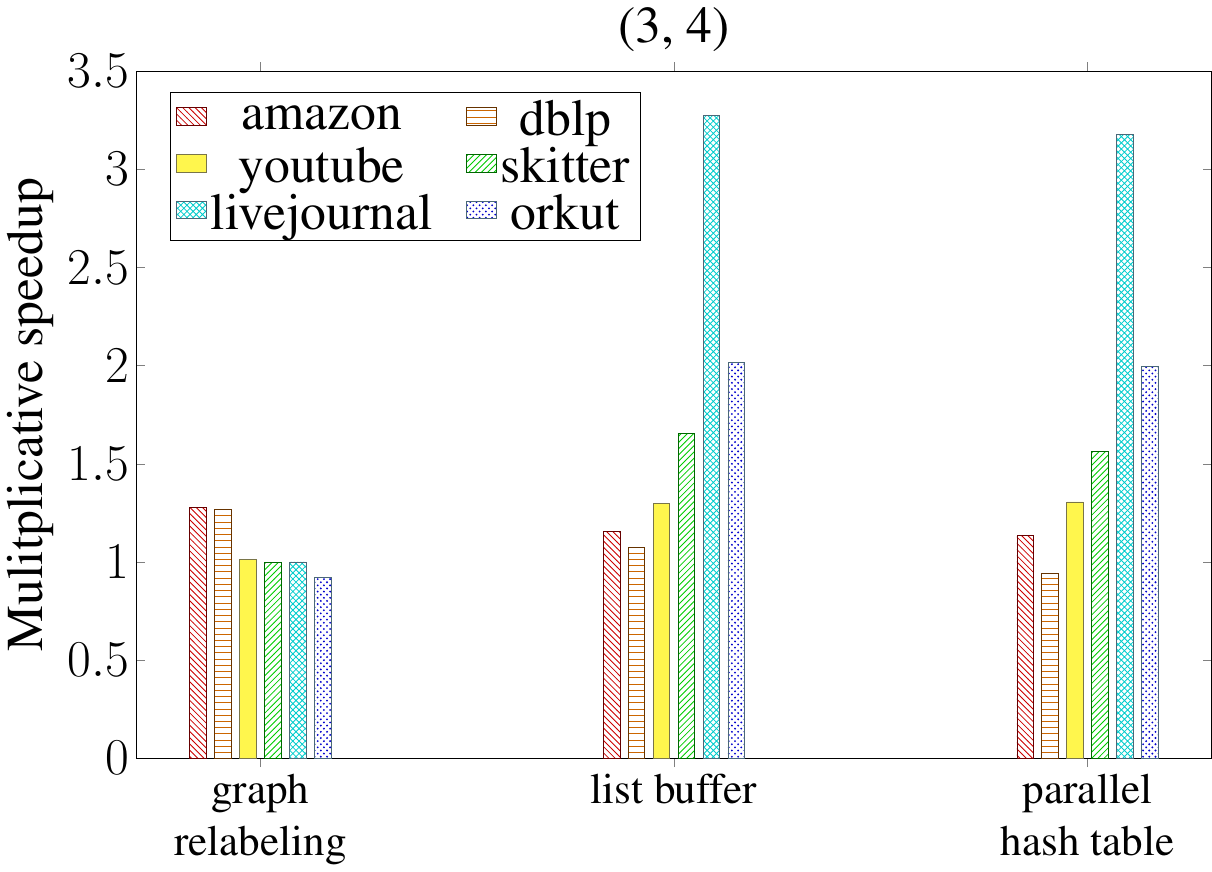}
      \end{subfigure}
   \caption{Multiplicative speedups of the graph relabeling, update aggregation, and graph contraction optimizations in \ournd{}, over a two-level setting with contiguous space and stored pointers, and using the simple array for $U$. Friendster is omitted from the $(2, 4)$ and $(3, 4)$ nucleus decomposition experiments, because the unoptimized \ournd{} times out for $(2,4)$, and \ournd{} runs out of memory for $(3,4)$.
   }
    \label{fig:nd-tune-other}
\end{figure*}

The most unoptimized form of \ournd{} uses a one-level parallel hash table $T$, no graph relabeling, and the simple array method for the update aggregation (as this is the simplest and most intuitive method), and in the case of $(2, 3)$ nucleus decomposition, no graph contraction. Note that the contiguous space and inverse index map optimizations do not apply to the one-level $T$.

Because the first three optimizations, namely numbers of levels, contiguous space, and the inverse index map, apply specifically to the parallel hash table $T$, while the remaining optimizations apply generally to the rest of the nucleus decomposition algorithm, 
we first tune different combinations of options for the first three optimizations against the unoptimized \ournd{}. We then separately tune the remaining optimizations, namely graph relabeling, update aggregation, and graph contraction optimizations, against both the unoptimized \ournd{} and the best choice of optimizations from the first comparison.

\begin{figure*}[t]
    \centering
   \begin{subfigure}{.53\textwidth}
   \centering
   \includegraphics[width=\columnwidth, page=4]{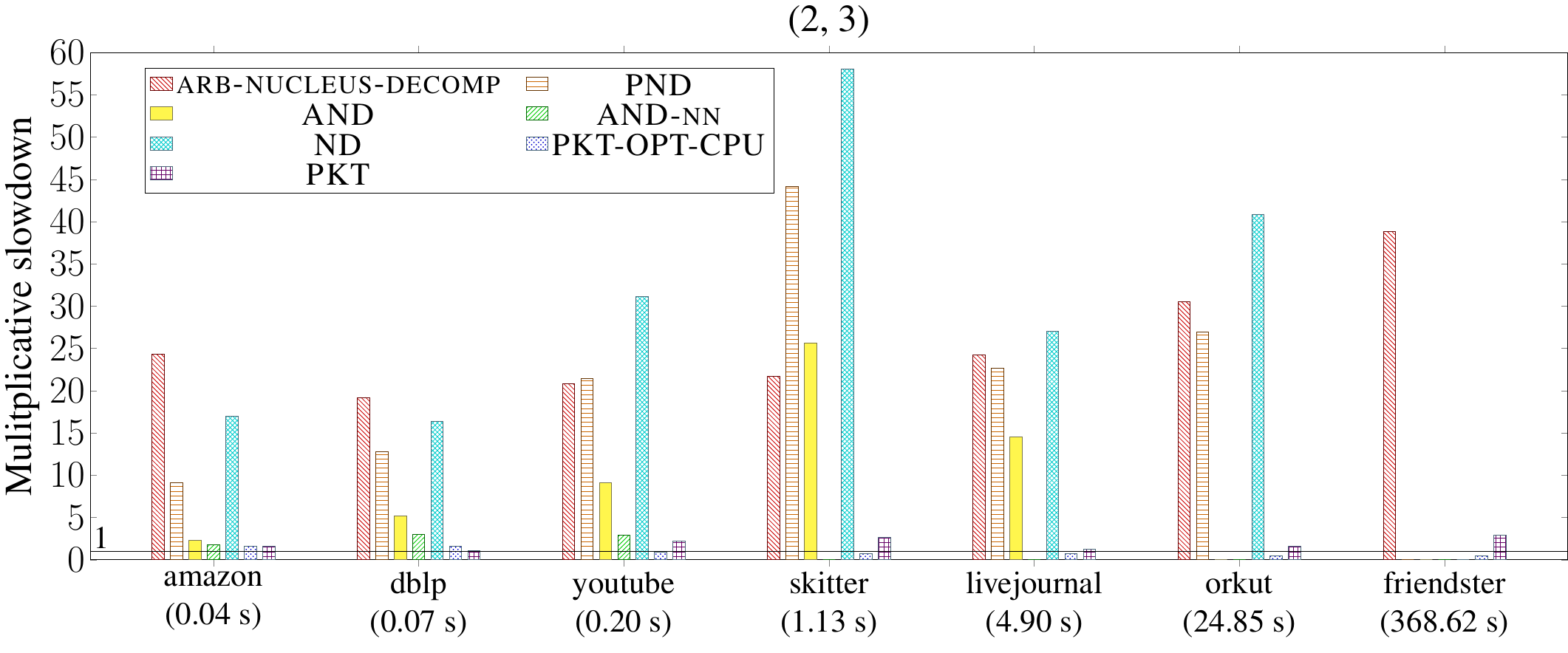}
   \end{subfigure}%
   \hfill
   \begin{subfigure}{.42\textwidth}
   \centering
   \includegraphics[width=\columnwidth, page=4]{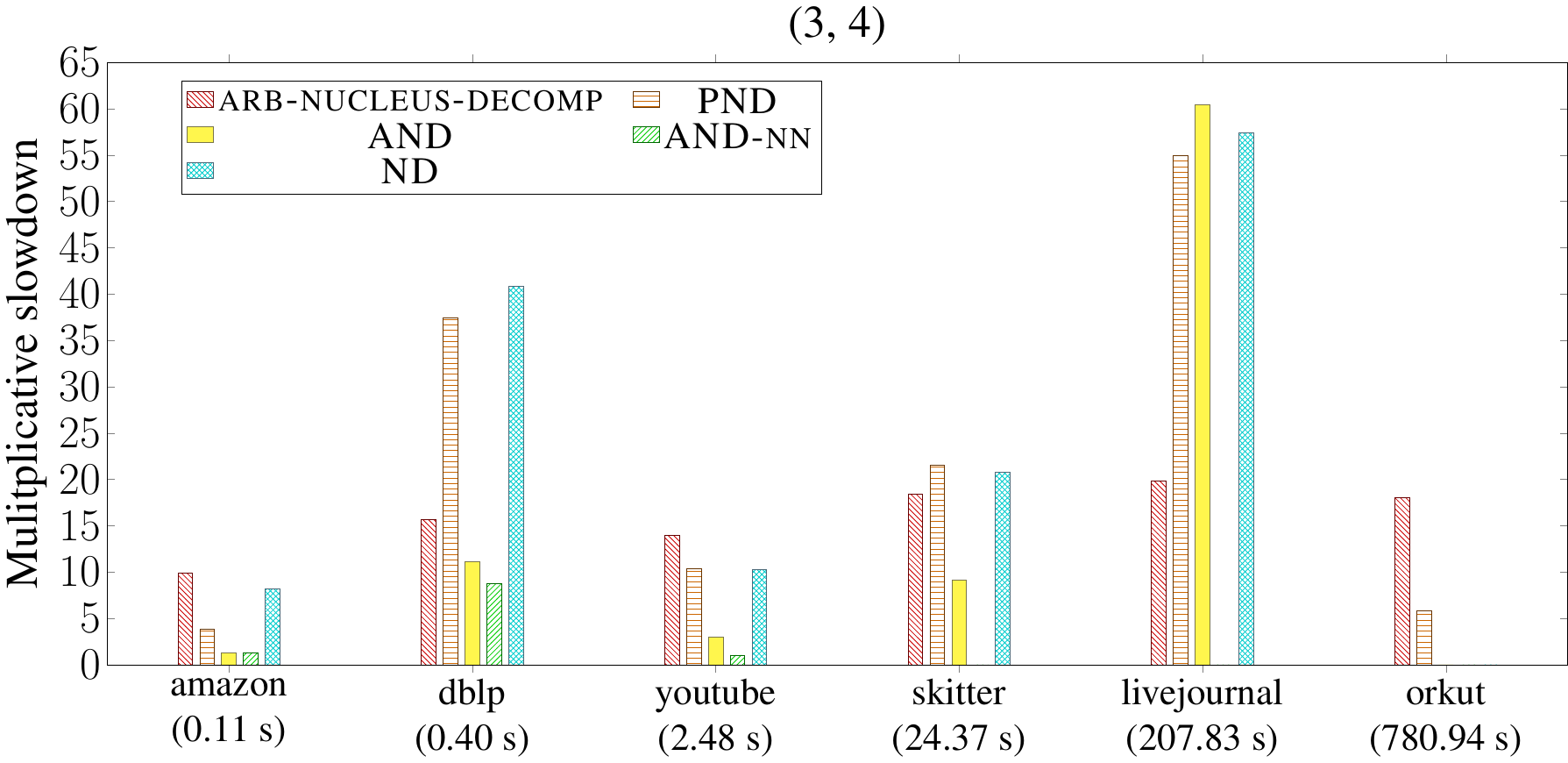}
      \end{subfigure}%
   \caption{{\color{change}Multiplicative slowdowns over our parallel \ournd{} of \algname{pkt-opt-cpu} and \algname{pkt} for $(2, 3)$ nucleus decomposition, and of single-threaded \ournd{}, \algname{pnd}, \algname{and}, \algname{and-nn}, and \algname{nd}
   for $(2, 3)$ and $(3, 4)$ nucleus decomposition. The label \ournd{} in the legend refers to our single-threaded running times. 
   We have omitted bars for \algname{pnd}, \algname{and}, \algname{and-nn}, and \algname{nd} where these implementations run out of memory or time out. We have included in parentheses the times of our parallel \ournd{} on 30 cores with hyper-threading. We have also included a line marking a multiplicative slowdown of 1 for $r=2$, $s=3$, and we see that \algname{pkt-opt-cpu} outperforms \ournd{} on skitter, livejournal, orkut, and friendster.}
   }
    \label{fig:speedups-23-34}
\end{figure*}

\begin{figure}[t]
   \centering
   \includegraphics[width=0.9\columnwidth, page=7]{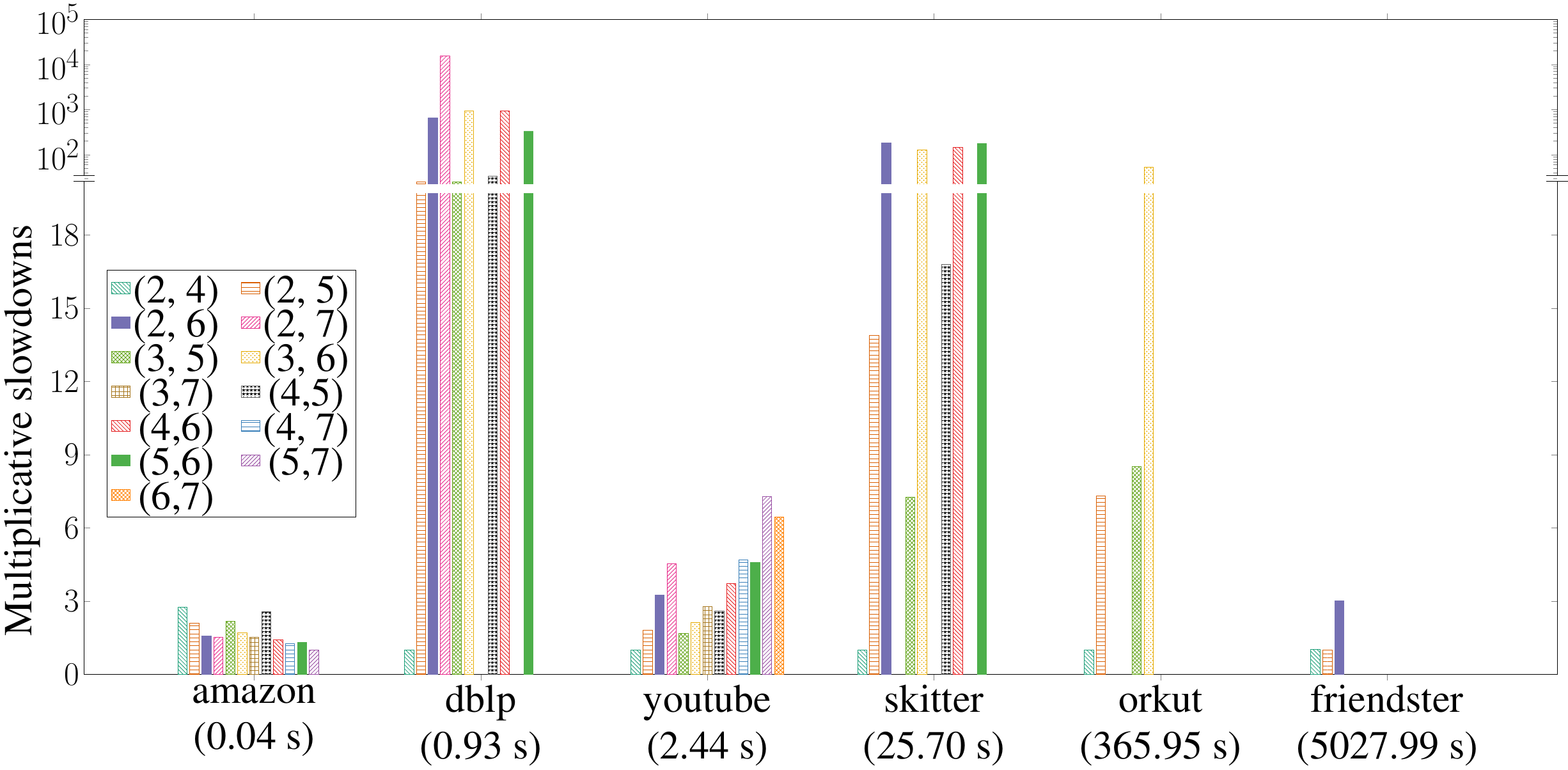}
   \caption{
   Multiplicative slowdowns of parallel \ournd{} for each $(r,s)$ combination over the fastest running time for parallel \ournd{} across all $r < s \leq 7$ for each graph (excluding $(2, 3)$ and $(3, 4)$, which are shown in Figure \ref{fig:speedups-23-34}). The fastest running time is labeled in parentheses below each graph. Also, livejournal is excluded because for these $r$ and $s$, \ournd{} is only able to complete $(2, 4)$ nucleus decomposition, in 484.74 seconds, and the rest timed out. We have omitted bars where \ournd{} runs out of memory or times out. 
   }
    \label{fig:speedups-rest}
\end{figure}

\begin{figure*}[t]
    \centering
    \begin{minipage}{.25\textwidth}
        \centering
        \includegraphics[width=0.9\columnwidth, page=11]{figures/fig_speedups.pdf}
   \caption{Speedup of \ournd{} over its single-threaded running times. {\color{change} "30h" denotes 30-cores with two-way hyper-threading.}
   }
    \label{fig:num-workers}
    \end{minipage}%
    \hspace{10pt}
    \begin{minipage}{0.72\textwidth}
        \centering
        \begin{subfigure}{.33\textwidth}
   \centering
   \includegraphics[width=\columnwidth, page=1]{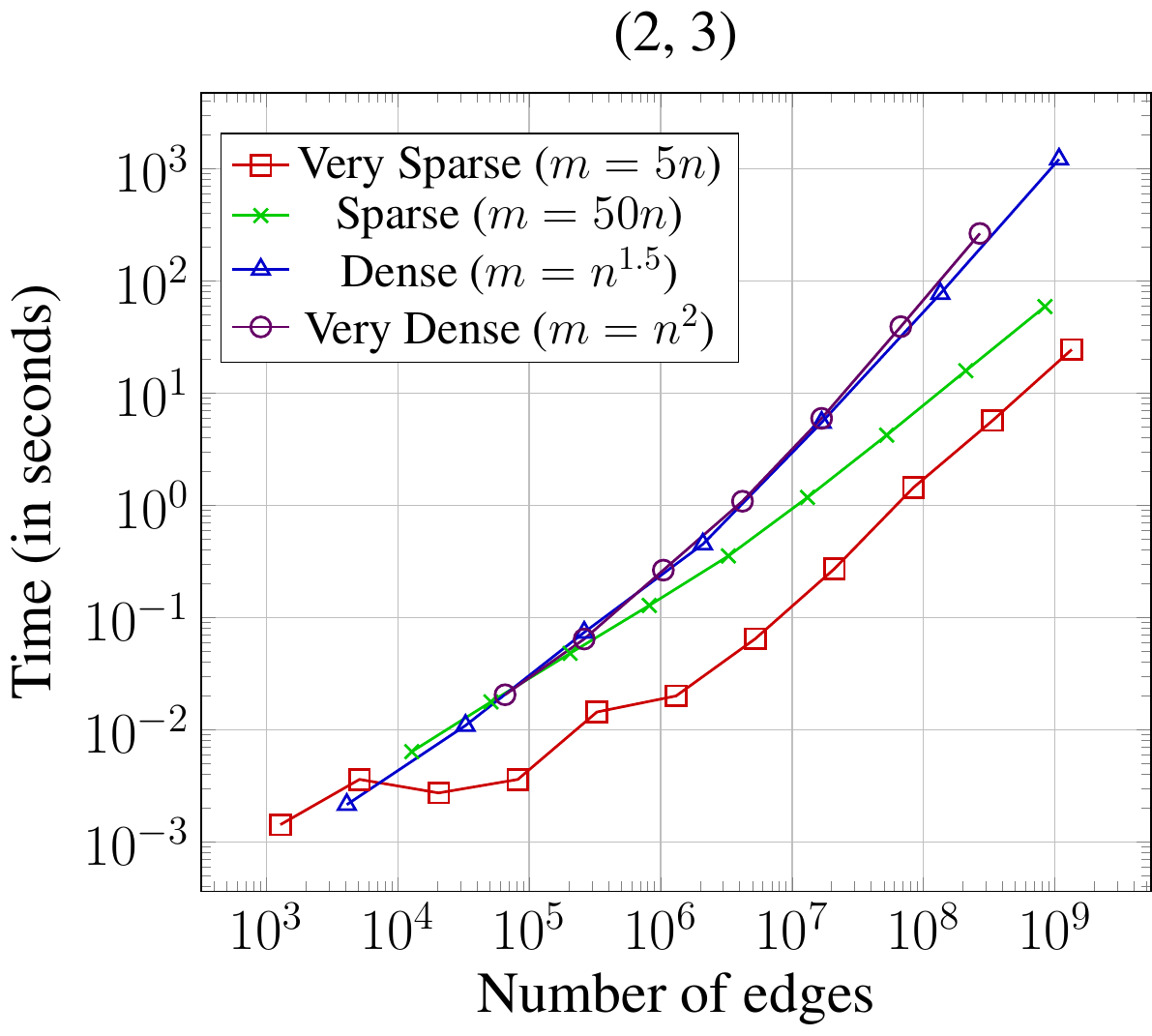}
   \end{subfigure}%
   \hfill
   \begin{subfigure}{.33\textwidth}
   \centering
   \includegraphics[width=\columnwidth, page=2]{figures/fig_revision.pdf}
      \end{subfigure}%
      \hfill
   \begin{subfigure}{.33\textwidth}
   \centering
   \includegraphics[width=\columnwidth, page=3]{figures/fig_revision.pdf}
      \end{subfigure}%
   \caption{\color{change} Running times of parallel \ournd{} on rMAT graphs of varying densities for $(2, 3)$, $(3, 4)$, and $(4, 5)$ nucleus decomposition. We remove duplicate generated edges.
   }
    \label{fig:rmat-23-34-45}
    \end{minipage}
\end{figure*}

\myparagraph{Optimizations on $T$}
Across $(2,3)$, $(2, 4)$, $(3, 4)$, and $(4, 5)$ nucleus decomposition, using a two-level $T$ with contiguous space and stored pointers for the inverse index map gives the overall best configuration across these $r$ and $s$ values, with up to 1.32x speedups over the unoptimized case. This configuration either outperforms or offers comparable performance to other configurations. Figure \ref{fig:nd-tune-lvl} and \ref{fig:nd-tune-lvl-45} shows the multiplicative speedup of each combination
of optimizations on $T$ over the unoptimized $T$ for $(3, 4)$ nucleus decomposition and $(4, 5)$ nucleus decomposition, respectively; friendster is omitted from our $(3, 4)$ experiments because \ournd{} runs out of memory on these graphs. Also, the speedups for $(2, 3)$ and $(2, 4)$ nucleus decomposition are omitted, but show similar behavior overall.

On the smallest graph, amazon, we see universally poor performance of the two-level and multi-level options, and the one-level $T$ outperforms these options; however, we note that the one-level running times for \ournd{} on amazon in these cases is $< 0.2$ seconds and the maximum core number of amazon is particularly small ($\leq 10$), and so the optimizations simply add too much overhead to see performance improvements.
The main case where the two-level $T$ performs noticeably worse than other options is for large graphs and for large $r$ and $s$. 
In $(3, 4)$ nucleus decomposition, we note that for orkut, using a $3$-multi-level $T$, also with contiguous space and stored pointers, significantly outperforms the analogous two-level option, with a 1.34x speedup over the unoptimized case compared to a 1.11x speedup. 
Similarly, in $(4, 5)$ nucleus decomposition, using a 3-multi-level $T$ offers comparable performance to the two-level $T$ on dblp and skitter (the 3-multi-level $T$ gives 1.46x and 1.06x speedups over the unoptimized $T$, respectively, while the 2-level $T$ gives 1.32x and 1.06x speedups, respectively), but underperforms on amazon and youtube. The benefit of using large $\ell$ relative to $r$ is difficult to observe for small $r$, since $\ell \leq r$ and speedups only appear in  graphs with sufficiently many $r$-cliques. 
We see relatively small speedups from larger $\ell$ for certain graphs, but in these cases, the performance is comparable, so we consider the two-level $T$ to be the best overall option.

Additionally, the two-level and multi-level options offer significant space savings due to their compact representation of vertices shared among many $r$-cliques, particularly for larger graphs and larger values of $r$ and $s$. 
Figures \ref{fig:nd-tune-lvl} and \ref{fig:nd-tune-space-45} also show the space savings in $T$ of each optimization on $(3, 4)$ nucleus decomposition and $(4, 5)$ nucleus decomposition, respectively; we omit the figures for $(2, 3)$ and $(2, 4)$ nucleus decomposition due to space limitations, although they show similar behavior. 
Across $(2, 3)$ and $(2, 4)$ nucleus decomposition, the two-level options give up to a 1.79x reduction
in space usage, and this increases to up to a 2.15x reduction in space usage on $(3, 4)$ nucleus decomposition and up to a 2.51x reduction in space usage on $(4, 5)$ nucleus decomposition. 
We similarly see greater space reductions in using $\ell$-multi-level $T$ for $\ell > 2$ on $(4, 5)$ nucleus decomposition compared to the same $\ell$ on $(3, 4)$ nucleus decomposition for the same graphs; however, $r$ is not large enough and there is not enough overlap between $r$-cliques such that using $\ell > 2$ offers significant space savings over the two-level options.

Overall, the optimal setting for the parallel hash table $T$  is a two-level combination of an array and a parallel hash table, with contiguous space and stored pointers for the inverse index map.

\myparagraph{Other optimizations}
We now consider the graph relabeling and update aggregation optimizations, fixing a one-level setting and a two-level setting with contiguous space and stored pointers for the inverse index map, and using the simple array for $U$ with a fetch-and-add to reserve every slot. Figure \ref{fig:nd-tune-other} shows these speedups for the two-level case; we omit the analogous figure for the one-level case, which shows similar behavior to the two-level case. 

Across $(2,3)$, $(2, 4)$, and $(3, 4)$ nucleus decomposition, the graph relabeling optimization gives up to 1.23x speedups on the one-level case and up to 1.29x speedups on the two-level case. We see greater speedups on the two-level case, because the increased locality across both levels due to the relabeling is more significant, whereas there is little improved locality in the one-level case. We note that graph relabeling provides minimal speedups in $(2,3)$ nucleus decomposition, with slowdowns of up to 1.11x, whereas graph relabeling is almost universally optimal in $(2, 4)$ and $(3, 4)$ nucleus decomposition. This is due to fewer benefits from locality when merely computing triangle counts from edges, versus computing higher clique counts.

In terms of the options for update aggregation, using a list buffer gives up to 3.43x speedups on the one-level case and up to 3.98x speedups on the two-level case. Using a parallel hash table gives up to 3.37x speedups on the one-level case and up to 4.12x speedups on the two-level case. Notably, the parallel hash table is the fastest for $(2, 3)$ nucleus decomposition, whereas the list buffer outperforms the parallel hash table for $(2, 4)$ and $(3, 4)$ nucleus decomposition. This behavior is particularly evident on larger graphs, where the time to compute updated $s$-clique counts is more significant, and thus there is less contention in using a list buffer.

Finally, in the special case of $(2,3)$ nucleus decomposition, we evaluate the performance of graph contraction over the two-level setting. We see up to 1.08x speedups using graph contraction, but up to 1.11x slowdowns when using graph contraction on small graphs, due to the increased overhead; however, the two-level running times of \ournd{} on these graphs is $<0.2$ seconds.

Overall, the optimal setting for $(2,3)$ nucleus decomposition is to use a parallel hash table for update aggregation and graph contraction (with no graph relabeling), and the optimal setting for general $(r, s)$ nucleus decomposition is to use a list buffer for update aggregation and graph relabeling. Combining all optimizations, over $(2, 3)$, $(2, 4)$, and $(3, 4)$ nucleus decomposition, we see up to a 5.10x speedup over the unoptimized \ournd{}.

\subsection{Performance}

Figures \ref{fig:speedups-23-34} and \ref{fig:speedups-rest} show the parallel runtimes for \ournd{} using the optimal settings described in Section \ref{sec:eval-tuning}, for $(r,s)$ where $r < s \leq 7$. 
We only show in these figures our self-relative speedups on $(r,s)=(2,3)$ and $(r,s)=(3,4)$, but we computed self-relative speedups for all $r < s \leq 7$, and  
overall, on 30 cores with two-way hyper-threading, \ournd{} obtains 3.31--40.14x self-relative speedups. We see larger speedups on larger graphs and for greater $r$. We also see good scalability over different numbers of threads, which we show in Figure \ref{fig:num-workers} for $(2, 3)$, $(2, 4)$, and $(3, 4)$ nucleus decomposition on dblp, skitter, and livejournal. {\color{change} Figure \ref{fig:rmat-23-34-45} additionally shows the scalability of \ournd{} over rMAT graphs of varying sizes and varying edge densities. We see that our algorithms scale in accordance with the increase in the number of $s$-cliques, depending on the density of the graph.} 

\myparagraph{Comparison to other implementations}
Figure \ref{fig:speedups-23-34} also shows the comparison of our parallel $(2, 3)$ and $(3, 4)$ nucleus decomposition implementations to other implementations.
We compare to \sariyuce  \textit{et al.}'s~\cite{Sariyuce2017,sariyuce2017parallel} parallel implementations, including their global implementation \algname{pnd}, their asynchronous local implementation \algname{and}, and their asynchronous local implementation with notification \algname{and-nn}, where the notification mechanism offers performance improvements at the cost of space usage. We run the local implementations to convergence. We furthermore compare to their implementation \algname{nd}, which is a serial version of \algname{pnd}.
We note that \sariyuce \textit{et al.} provide implementations only for $(2, 3)$ and $(3, 4)$ nucleus decomposition. 

Compared to \algname{pnd}, \ournd{} achieves 3.84--54.96x speedups, and compared to \algname{and}, \ournd{} achieves 1.32--60.44x speedups. Notably, \algname{pnd} runs out of memory on friendster for $(2, 3)$ nucleus decomposition, while our implementation can process friendster in 368.62 seconds. Moreover, \algname{and} runs out of memory on both orkut and friendster for both $(2, 3)$ and $(3, 4)$ nucleus decomposition, while \ournd{} is able to process orkut and friendster for $(2, 3)$ nucleus decomposition, and orkut for $(3, 4)$ nucleus decomposition.

Compared to \algname{and-nn}, \ournd{} achieves 1.04--8.78x speedups. \algname{and-nn} outperforms the other implementations by \sariyuce  \textit{et al.}, but due to its increased space usage, it is unable to run on the larger graphs skitter, livejournal, orkut, and friendster for both $(2,3)$ and $(3,4)$ nucleus decomposition.
Considering the best of \sariyuce \textit{et al.}'s parallel implementations for each graph and each $(r,s)$, \ournd{} achieves 1.04--54.96x speedups overall. Compared to \sariyuce \textit{et al.}'s serial implementation \algname{nd}, \ournd{} achieves 8.19--58.02x speedups. We significantly outperform \sariyuce \textit{et al.}'s algorithms due to the work-efficiency of our algorithm.
{\color{change} Notably, \algname{and} and \algname{and-nn} are not work-efficient because they perform a local algorithm, in which each $r$-clique locally updates its $s$-clique-core number until convergence, compared to the work-efficient peeling process where the minimum $s$-clique-core is extracted from the entire graph in each round. The total work of these local updates can greatly outweigh the total work of the peeling process. We measured the total number of times $s$-cliques were discovered in each algorithm, and found that \algname{and} computes 1.69--46.03x the number of $s$-cliques in \ournd{}, with a median of 15.15x. \algname{and-nn} reduces this at the cost of space, but still computes up to 3.45x the number of $s$-cliques in \ournd{}, with a median of 1.4x.

Moreover, \algname{pnd} does perform a global peeling-based algorithm like \ournd{}, but does not parallelize within the peeling process; more concretely, all $r$-cliques with the same $s$-clique count can be peeled simultaneously, which \ournd{} accomplishes by introducing optimizations, notably the update aggregation optimization, that specifically address synchronization issues when peeling multiple $r$-cliques simultaneously. \algname{pnd} instead peels these $r$-cliques sequentially in order to avoid these synchronization problems, leading to significantly more sequential peeling rounds than required in \ournd{}; in fact, \algname{pnd} performs 5608--84170x the number of rounds of \ournd{}. }

{\color{change} We note that additionally, the speedups of \ournd{} over \algname{pnd}, \algname{and}, and \algname{and-nn} are not solely due to our use of an efficient parallel $k$-clique counting subroutine~\cite{shi2020parallel}. We replaced the $k$-clique counting subroutine in \ournd{} with that used by \sariyuce  \textit{et al.}~\cite{Sariyuce2017,sariyuce2017parallel}, and found that \ournd{} using \sariyuce  \textit{et al.}'s $k$-clique counting subroutine achieves between 1.83--28.38x speedups over \sariyuce  \textit{et al.}'s  best implementations overall for $(2, 3)$ and $(3, 4)$ nucleus decomposition. Within \ournd{}, the efficient $k$-clique counting subroutine gives up to 3.04x speedups over the subroutine used by \sariyuce  \textit{et al.}, with a median speedup of 1.03x.}

\myparagraph{Comparison to $k$-truss implementations}
In the special case of $(2, 3)$ nucleus decomposition, or $k$-truss, we compare our parallel implementation to Che \textit{et al.}'s~\cite{ChLaSuWaLu20} highly optimized parallel CPU implementation \algname{pkt-opt-cpu}, using all of their successive optimizations and considering the best performance across different reordering options, including degree reordering, $k$-core reordering, and no reordering. Figure \ref{fig:speedups-23-34} also shows the comparison of \ournd{} to \algname{pkt-opt-cpu}. \ournd{} achieves up to 1.64x speedups on small graphs, but up to 2.27x slowdowns on large graphs compared to \algname{pkt-opt-cpu}. However, we note that \algname{pkt-opt-cpu} is limited in that it solely implements $(2, 3)$ nucleus decomposition, and its methods do not generalize to other values of $(r,s)$. 
{\color{change} \ournd{} outperforms \algname{pkt-opt-cpu} on small graphs due to a more efficient graph reordering subroutine; \ournd{} computes a low out-degree orientation which it then uses to reorder the graph, and \ournd{}'s reordering subroutine achieves a 3.07--5.16x speedup over \algname{pkt-opt-cpu}'s reordering subroutine.\footnote{\color{change} For \algname{pkt-opt-cpu}, we consider the reordering option that gives the fastest overall running time for each graph.} \algname{pkt-opt-cpu} uses its own parallel sample sort implementation, which is slower compared to that used by \ournd{}~\cite{DhBlSh18}. However, the cost of graph reordering is negligible compared to the cost of computing the $(2, 3)$ nucleus decomposition in large graphs, and \algname{pkt-opt-cpu} uses highly optimized intersection subroutines which achieve greater speedups over \ournd{}'s generalized implementation.}

We also compared to additional $(2, 3)$ nucleus decomposition implementations. Specifically, we compared to 
Kabir and Madduri's \algname{pkt}~\cite{kabir2017parallel}, which outperforms Che \textit{et al.}'s \algname{pkt-opt-cpu} on the small graphs amazon and dblp by 1.01--1.53x, and which is also shown in Figure \ref{fig:speedups-23-34}. However, \ournd{} achieves 1.07--2.88x speedups over \algname{pkt} on all graphs, including amazon and dblp. Moreover, we compared to Smith \textit{et al.}'s \algname{msp}~\cite{smith2017truss}, but found that \algname{msp} is slower than \algname{pkt} and \algname{pkt-opt-cpu}, and \ournd{} achieves 2.35--7.65x speedups over \algname{msp}.
We compared to Blanco \textit{et al.}'s~\cite{BlLoKi19} implementations as well, which we found to also be slower than \algname{pkt} and \algname{pkt-opt-cpu}, and \ournd{} achieves 2.45--21.36x speedups over their best implementation.\footnote{Blanco \textit{et al.}'s implementation requires the maximum $k$-truss number to be provided as an input, which we did in these experiments.}
We also found that Che \textit{et al.}'s implementations outperform the reported numbers from a recent parallel $(2,3)$ nucleus decomposition implementation by Conte \textit{et al.}~\cite{CoDeGrMaVe18}.

\section{Conclusion}
We have presented a novel theoretically efficient parallel algorithm
for $(r, s)$ nucleus decomposition, which
improves upon the previous best
theoretical bounds. We have also developed practical
optimizations and showed that they significantly improve the
performance of our algorithm. Finally, we have provided a comprehensive
experimental evaluation demonstrating that on a 30-core machine with two-way hyper-threading 
our algorithm achieves 
up to 55x speedup over the previous state-of-the-art parallel
implementation.

\begin{acks}
This research was supported by NSF Graduate Research Fellowship
\#1122374, 
DOE Early Career Award \#DE-SC0018947,
NSF CAREER Award \#CCF-1845763, Google Faculty Research Award, Google Research Scholar Award, DARPA
SDH Award \#HR0011-18-3-0007, and Applications Driving Architectures
(ADA) Research Center, a JUMP Center co-sponsored by the Semiconductor Research Corporation (SRC) and DARPA.
\end{acks}

\bibliographystyle{ACM-Reference-Format}
\bibliography{references}

\appendix 

\section{Discussion of Theoretically Efficient Bounds}
We discuss in this section the proof that the complexity bounds for our $(r, s)$ nucleus decomposition algorithm, \ournd{}, improve upon that in prior work. In particular, \ournd{} takes space proportional to the number of $r$-cliques in the graph, and we proved the following work bounds for \ournd{} in Section \ref{sec:alg-nd}, where $\rho_{(r, s)}(G)$ is upper bounded by the total number of $r$-cliques in the graph by definition.

\begin{reptheorem}{thm:peelexact}
  \ournd{} computes the $(r,s)$ nucleus decomposition in $\BigO{m\alpha^{s-2} +
  \rho_{(r,s)}(G)\log n}$ amortized expected work and $\BigO{\rho_{(r,s)}(G) \log n + \log^2 n}$ span \whp{}, where
  $\rho_{(r,s)}(G)$ is the $(r,s)$ peeling complexity of $G$.
\end{reptheorem}

We compare our bound to those given by \sariyuce{} \textit{et al.}~\cite{Sariyuce2017} for their sequential $(r, s)$ nucleus decomposition algorithms. Their bounds are given in terms of the number of $c$-cliques containing each vertex $v$, or $ct_c(v)$, and the work of an arbitrary $c$-clique enumeration algorithm, or $RT_c$. We also let $n_c$ denote the total number of $c$-cliques.
In particular, they give two complexity bounds depending on the space usage of their algorithms.
Assuming space proportional to the number of $s$-cliques and the number of $r$-cliques in the graph, they compute the $(r,s)$ nucleus decomposition in $\BigO{RT_r + RT_s}=\BigO{RT_s}$ work, and assuming space proportional to only the number of $r$-cliques in the graph, they compute the $(r,s)$ nucleus decomposition in $\BigO{RT_r + \sum_{v \in V(G)} ct_r(v) \cdot \text{deg}(v)^{s-r}}$ work.

Assuming a work-efficient $c$-clique listing algorithm, specifically Shi \textit{et al.}'s work-efficient $c$-clique listing algorithm~\cite{shi2020parallel}, we see that \sariyuce{} \textit{et al.}'s bounds are $\BigO{m \alpha^{s-2}}$ and $\BigO{m \alpha^{r-2} + \sum_{v \in V(G)} ct_r(v) \cdot \text{deg}(v)^{s-r}}$ assuming space proportional to the number of $s$-cliques and $r$-cliques and space proportional to only the number of $r$-cliques, respectively.

We note that \sariyuce{} \textit{et al.}\ do not include the work required to retrieve the $r$-clique with the minimum $s$-clique count in their complexity bound. Importantly, the time complexity of this step is non-negligible in the theoretical bounds of $(r, s)$ nucleus decomposition. There are two possibilities depending on space usage. 
If space proportional to the number of $s$-cliques is allowed, it is possible to use an array proportional to the maximum number of $s$-cliques per vertex to allow for efficient retrieval of the $r$-clique with the minimum $s$-clique count in any given round, which would take $\BigO{n_s}$ work in total.
However, if the space is limited to the number of $r$-cliques, a space-efficient heap must be used, 
which would add an additional $\BigO{n_r \log n}$ term to the work. 

Adding the time complexity of retrieving the $r$-clique with the minimum $s$-clique count, assuming space proportional to the number of $s$-cliques and $r$-cliques, \sariyuce{} \textit{et al.}'s algorithm takes $\BigO{m \alpha^{s-2} + n_s} = \BigO{m \alpha^{s-2}}$ work, and assuming space proportional to only the number of $r$-cliques, \sariyuce{} \textit{et al.}'s algorithm takes $\BigO{m \alpha^{r-2} + \sum_{v \in V(G)} ct_r(v) \cdot \text{deg}(v)^{s-r} + n_r \log n}$ work. 

\myparagraph{Space proportional to the number of $s$-cliques and $r$-cliques}
We first compare to \sariyuce{} \textit{et al.}'s work considering space proportional to the number of $s$-cliques and $r$-cliques. 

We note that assuming space proportional to the number of $s$-cliques, we can replace the $\BigO{\rho_{(r,s)}(G)\log n}$ term in \ournd{}'s work bound with $\BigO{n_s}$ while maintaining the same span by
replacing the batch-parallel Fibonacci heap~\cite{Shi2020} with an array of size proportional to the total number of $s$-cliques (to store and retrieve the $r$-cliques with the minimum $s$-clique count in any given round).

In more detail, let our bucketing structure be represented by an array of size $x = \BigO{\text{poly}(y)}$, where each array cell represents a bucket, and there are $y$ elements in the bucket structure. Then, we can
repeatedly pop the minimum bucket until the structure is empty using $\BigO{x}$ total work and $\BigO{\rho \log y}$ span, assuming it takes $\BigO{\rho}$ rounds to empty the structure. This is done by splitting the array into regions $r_i = [2^i, 2^{i + 1}]$ for $i = [0, \log x)$. Let us denote the number of buckets in each region $r_i$ by $|r_i|$. We start at the first region $r_0$ and search in parallel for the first non-empty bucket, using a parallel reduce. We return the first non-empty bucket if it exists, and otherwise, we repeat with the next region $r_1$, and so on. We note that we will never revisit the previously searched region $r_0$ once it has been found to be empty. We repeat the entire process of splitting the array into regions and searching each region in order for each subsequent pop query, starting from the previously popped bucket. In this manner, we maintain  $\BigO{\log y}$ span to pop each non-empty bucket, while incurring $\BigO{x}$ total work.

Thus, allowing space proportional to the number of $s$-cliques, \ournd{} is work-efficient, taking $\BigO{m \alpha^{s-2} + n_s} = \BigO{m \alpha^{s-2}}$ work.

\myparagraph{Space proportional to the number of $r$-cliques}
We now discuss \sariyuce{} \textit{et al.}'s work considering space proportional to only the number of $r$-cliques.
We claim that \ournd{} improves upon \sariyuce{} \textit{et al.}'s algorithm.

In order to show this claim, we discuss the work of each step of \ournd{} in more detail, as we do in the proof of Theorem \ref{thm:peelexact}.\footnote{We provide a closed form work bound for \ournd{}, whereas \sariyuce{} \textit{et al.}'s bound is not closed. As a result, directly comparing our closed form bound to \sariyuce{} \textit{et al.}'s bound is difficult; instead, we provide here a more detailed comparison of the exact work required by \ournd{} with \sariyuce{} \textit{et al.}'s work bound.}
First, we consider the work of inserting $r$-cliques into our bucketing structure $B$, which takes $\BigO{m \alpha^{r-2}}$ work. Then, because each $r$-clique has its bucket decremented at most once per incident $s$-clique, we incur work proportional to the total number of $s$-cliques, or $\BigO{ct_s(v)}$. Extracting the minimum bucket takes $\BigO{\rho_{(r,s)}(G)\log n }$ amortized expected work using the batch-parallel Fibonacci heap.

Finally, it remains to discuss the work of obtaining updated $s$-clique counts after peeling each set of $r$-cliques, in the \algname{Update} subroutine. As discussed in the proof of Theorem \ref{thm:peelexact}, it takes $\BigO{\sum_R \min_{1 \leq i \leq r} \text{deg}(v_i)}$ work \whp{} to intersect the neighbors of each vertex $v\in R$ over all $r$-cliques $R$. \ournd{} then considers the set of vertices in this intersection $I$, and performs successive intersection operations with the oriented neighbors of each vertex in $I$. These intersection operations are bounded by the arboricity $\BigO{\alpha}$, but they are also bounded by $\BigO{\min_{1 \leq i \leq r} \text{deg}(v_i)}$ as the maximum size of $I$. Thus, this step takes $\BigO{\sum_R \min_{1 \leq i \leq r} \allowbreak  \text{deg}(v_i)^{s-r}} \allowbreak = \allowbreak  \BigO{\sum_{v \in V(G)} ct_r(v) \cdot \text{deg}(v)^{s-r}}$ work.

In total, \ournd{} incurs as an upper bound $\BigO{m \alpha^{r-2} + ct_s(v) + \sum_{v \in V(G)} ct_r(v) \cdot \text{deg}(v)^{s-r} +  \rho_{(r,s)}(G)\log n } = 
\BigO{m \alpha^{r-2} + \sum_{v \in V(G)} ct_r(v) \cdot \text{deg}(v)^{s-r} +  \rho_{(r,s)}(G)\log n }$ work, and thus does not exceed \sariyuce{} \textit{et al.}'s work bound.

Note that \ournd{} improves upon \sariyuce{} \textit{et al.}'s work in the \algname{Update} subroutine. Notably, \sariyuce{} \textit{et al.}'s work assumes that for each vertex $v$ participating in an $r$-clique, the work incurred is $\BigO{\text{deg}(v)^{s-r}}$. However, our \algname{Update} subroutine considers only the minimum vertex degree of each $r$-clique. Moreover, after finding the set $I$ of candidate vertices to extend each $r$-clique into an $s$-clique, \ournd{} uses the $\BigO{\alpha}$-oriented neighbors of these candidate vertices in the intersection subroutine. The actual work incurred in the intersection is thus the minimum of the size of $I$ and of the arboricity oriented out-degrees of the vertices in $I$, which further improves upon the $\BigO{\text{deg}(v)^{s-r}}$ bound in \sariyuce{} \textit{et al.}'s work. This improvement is particularly evident if there are a few vertices of high degree involved in many $r$- and $s$-cliques, in which the use of an $\BigO{\alpha}$-orientation significantly reduces the number of out-neighbors that must be traversed. We additionally note that $\BigO{\alpha}$ tightly bounds the maximum out-degree in any acyclic orientation of a graph $G$,\footnote{This follows because if we let $d$ denote the degeneracy of the graph $G$, then the arboricity $\alpha$ tightly bounds the degeneracy in that $\alpha \leq d \leq 2 \alpha - 1$. Moreover, $d$ is the degeneracy of $G$ if and only if $G$ can be acyclically directed such that the maximum out-degree in the directed graph is $d$~\cite{Chrobak1991PlanarOW}.} and in this sense, our closed-form work bound for \ournd{} is never asymptotically worse than the work bound using any other acyclic orientation.

\end{document}